\newif\iffull
\fulltrue
\iffull
\documentclass[a4paper,UKenglish]{article}
\usepackage{fullpage}
\usepackage[hidelinks]{hyperref}
\else
\documentclass[a4paper,UKenglish]{socg-lipics-v2019}
\fi
\usepackage{amssymb, amsmath, amsthm}

\usepackage[utf8]{inputenc} 
\usepackage{graphicx}
\usepackage{xcolor}
\usepackage{wrapfig}
\usepackage{bbm}

\usepackage{makecell}

\usepackage{multirow}

\title{Online Packing to Minimize Area or Perimeter}
\iffull
\author{Mikkel Abrahamsen\thanks{Basic Algorithms Research Copenhagen (BARC), University of Copenhagen. BARC is supported by the VILLUM Foundation grant 16582.
Lorenzo Beretta received funding
from the European Union's Horizon 2020 research and innovation
program under the Marie Skłodowska-Curie grant agreement No.~801199.}\quad \quad Lorenzo Beretta$^*$
}
\else
\author{Mikkel Abrahamsen}{Basic Algorithms Research Copenhagen (BARC), University of Copenhagen.}{}{}{}
\author{Lorenzo Beretta}{Basic Algorithms Research Copenhagen (BARC), University of Copenhagen.}{}{}{}
\fi

\date{January 21, 2021}

\iffull
\else
\authorrunning{Mikkel Abrahamsen \& Lorenzo Beretta}

\ccsdesc{Theory of computation~Computational geometry}
\ccsdesc{Theory of computation~Packing and covering problems}
\ccsdesc{Theory of computation~Online algorithms}

\keywords{Packing, online algorithms}

\acknowledgements{
\begin{wrapfigure}{l}{0.055\textwidth}
\vspace{-4 mm}
\includegraphics[height=0.055\textwidth, width=0.055\textwidth]{barc.png}
\includegraphics[height=0.03\textwidth, width=0.055\textwidth]{euh2020.jpg}
\end{wrapfigure}
Research of both authors partly
supported by Investigator Grant 16582, Basic Algorithms Research Copenhagen (BARC), from the VILLUM Foundation.
Lorenzo Beretta received funding
from the European Union's Horizon 2020 research and innovation
program under the Marie Skłodowska-Curie grant agreement No 801199.
}
\fi

\newcommand{\comm}[1]{}

\newcommand{\R}{\mathbb{R}}
\newcommand{\N}{\mathbb{N}}
\newcommand{\Z}{\mathbb{Z}}

\newcommand{\D}{\mathcal{D}}

\newcommand{\eps}{\varepsilon}
\renewcommand{\phi}{\varphi}
\newcommand{\mydef}{:=}%should maybe more fancy, but now easy to change.

\newcommand{\OPT}{\textsc{Opt}}
\newcommand{\ALG}{\textsc{Alg}}
\newcommand{\arearot}{\textsc{Area\-Ro\-ta\-tion}}
\newcommand{\areatrans}{\textsc{Area\-Trans\-la\-tion}}
\newcommand{\perirot}{\textsc{Peri\-meter\-Ro\-ta\-tion}}
\newcommand{\peritrans}{\textsc{Peri\-meter\-Trans\-la\-tion}}
\newcommand{\sqArea}{\textsc{Square\-In\-Square\-Area}}

\newcommand{\NFS}{\textsc{NFS}}
\newcommand{\smallboxtrans}{\textsc{BrickTranslation}}
\newcommand{\smallboxrot}{\textsc{BrickRotation}}

\newcommand{\dynamicboxtrans}{\textsc{DynBoxTrans}}
\newcommand{\dynamicboxrot}{\textsc{DynBoxRot}}
\newcommand{\dynamicboxrotopt}{\textsc{DynBoxRot}\ensuremath{_{\sqrt[4]{\OPT}}}}
\newcommand{\dynamicboxrotmin}{\textsc{DynBoxRot}\ensuremath{_{\sqrt{n} \, \wedge \sqrt[4]{\OPT}}}}
\newcommand{\alg}{\textsc{Alg}}
\newcommand{\comprattrans}{4}
\newcommand{\compratrot}{4}
\newcommand{\compratsq}{6}

\newcommand{\spl}{\dagger}
\newcommand{\case}[2]{\noindent\textbf{Case (#1)} [\emph{#2}]}

\newcommand{\pparagraph}[1]{\paragraph*{#1}}

\iffull
\newtheorem{theorem}{Theorem}
\newtheorem{lemma}[theorem]{Lemma}

\newtheorem{claim}[theorem]{Claim}

\newtheorem{corollary}[theorem]{Corollary}

\theoremstyle{definition}

\newtheorem{remark}[theorem]{Remark}
\fi

\begin{document}

	\maketitle
\begin{abstract}
We consider online packing problems where we get a stream of axis-parallel rectangles.
The rectangles have to be placed in the plane without overlapping, and each rectangle must be placed without knowing the subsequent rectangles.
The goal is to minimize the perimeter or the area of the axis-parallel bounding box of the rectangles.
We either allow rotations by $90^\circ$ or translations only.

For the perimeter version we give algorithms with an absolute competitive ratio slightly less than $\comprattrans$ when only translations are allowed and when rotations are also allowed.
%With a slight modification, the same algorithm achieves the absolute competitive ratio of $\compratrot$ also when we allow rotations. 

We then turn our attention to minimizing the area and show that the asymptotic competitive ratio of any algorithm is at least $\Omega(\sqrt{n})$, where $n$ is the number of rectangles in the stream, and this holds with and without rotations.
We then present algorithms that match this bound in both cases and the competitive ratio is thus optimal to within a constant factor.
We also show that the competitive ratio cannot be bounded as a function of $\OPT$.
We then consider two special cases.

The first is when all the given rectangles have aspect ratios bounded by some constant.
% Then we give an algorithm with an absolute competitive ratio bounded by a constant.
The particular variant where all the rectangles are squares and we want to minimize the area of the bounding square has been studied before and an algorithm with a competitive ratio of $8$ has been given~[Fekete and Hoffmann, Algorithmica, 2017].
We improve the analysis of the algorithm and show that the ratio is at most $\compratsq$, which is tight.
% In particular, when all the rectangles are squares, we get a ratio of $\compratsq$.

The second special case is when all edges have length at least $1$.
Here, the $\Omega(\sqrt n)$ lower bound still holds, and we turn our attention to lower bounds depending on $\OPT$.
We show that any algorithm for the translational case has an asymptotic competitive ratio of at least $\Omega(\sqrt{\OPT})$.
If rotations are allowed, we show a lower bound of $\Omega(\sqrt[4]{\OPT})$.
For both versions, we give algorithms that match the respective lower bounds:
With translations only, this is just the algorithm from the general case with competitive ratio $O(\sqrt n)=O(\sqrt{\OPT})$.
If rotations are allowed, we give an algorithm with competitive ratio $O(\min\{\sqrt n,\sqrt[4]{\OPT}\})$, thus matching both lower bounds simultaneously.
\end{abstract}

\section{Introduction}
Problems related to packing appear in a plethora of big industries.
For instance, two-dimensional versions of packing arise when a given set of pieces have to be cut out from a large piece of material so as to minimize waste.
This is relevant to clothing production where cutting patterns are cut out from a roll of fabric, and similarly in leather, glass, wood, and sheet metal cutting.

In some applications, it is important that the pieces are placed in an \emph{online} fashion.
This means that the pieces arrive one by one and we need to decide the placement of one piece before we know the ones that will come in the future.
This is in contrast to \emph{offline} problems, where all the pieces are known in advance.
Problems related to packing were some of the first for which online algorithms were described and analyzed.
Indeed, the first use of the terms ``online'' and ``offline'' in the context of approximation algorithms was in the early 1970s and used for algorithms for bin-packing problems~\cite{Fiat1998}.

In this paper, we study online packing problems where the pieces can be placed anywhere in the plane as long as they do not overlap.
The goal is to minimize the region occupied by the pieces.
The pieces are axis-parallel rectangles, and they may or may not be rotated by $90^\circ$.
We want to minimize the size of the axis-parallel bounding box of the pieces, and the size of the box is either the perimeter or the area.
This results in four problems: $\perirot$, $\peritrans$, $\arearot$, and $\areatrans$.

\pparagraph{Competitive analysis}
The \emph{competitive ratio} of an online algorithm is the equivalent of the \emph{approximation ratio} of an (offline) approximation algorithm.
The usual definitions~\cite{borodin2005online,CHRISTENSEN201763,Csirik1998} of competitive ratio (or \emph{worst case ratio}, as it may also be called~\cite{Csirik1998}) can only be used to describe that the cost of the solution produced by an online algorithm is at most some constant factor higher than the cost $\OPT$ of the optimal (offline) solution.
In the study of approximation algorithms, it is often the case that the approximation ratio is described not just as a constant, but as a more general function of the input.
In the same way, we generalize the definition of competitive ratios to support such statements about online algorithms.

Consider an algorithm $A$ for one of the packing problems studied in this paper.
Let $\mathcal L$ be the set of non-empty streams of rectangular pieces.
For a stream $L\in\mathcal L$, we define $A(L)$ to be the cost of the packing produced by $A$ and let $\OPT(L)$ be the cost of the optimal (offline) packing.
We say that $A$ has an \emph{absolute competitive ratio} of $f(L)$, for some function $f:\mathcal L\longrightarrow\R^+$ which may just be a constant, if
\[
\sup_{L\in\mathcal L}\frac{A(L)}{\OPT(L)f(L)}\leq 1.
\]
We say that $A$ has an \emph{asymptotic competitive ratio} of $f(L)$ if
\[
\limsup_{c\longrightarrow\infty}\left(\sup\left\{\frac{A(L)}{\OPT(L)f(L)}\mid L\in\mathcal L\text{ and }\OPT(L)=c\right\}\right)\leq 1.
\]

% With slight abuse of terminology, we say that an algorithm $A$ has a competitive ratio of \emph{at least} $\Omega(f(L))$ if for any function $g(L)$ such that $A$ has an competitive ratio of $g(L)$ it holds that $g(L)=\Omega(f(L))$.
In this paper, the functions $f(L)$ that we consider will be (i) constants, (ii) functions of the number of pieces $n=|L|$, (iii) functions of $\OPT(L)$.

By definition, if $A$ has an absolute competitive ratio of $f(L)$, then $A$ also has an asymptotic competitive ratio of $f(L)$, but $A$ may also have a smaller asymptotic competitive ratio $g(L)<f(L)$.
\iffull
However, the following easy lemma shows that for the problems studied in this paper, any constant asymptotic competitive ratio can be matched to within an arbitrarily small difference by an absolute competitive ratio.

\begin{lemma}
For the problems studied in this paper, if an algorithm $A$ has an asymptotic competitive ratio of some constant $c>1$, then for every $\eps>0$, there is an algorithm $A'$ with absolute competitive ratio $c+\eps$.
It follows that any constant lower bound on the absolute competitive ratio is also a lower bound on the asymptotic competitive ratio.
\end{lemma}

\begin{proof}
Let $n>0$ be so large that when $\OPT(L)\geq n$, we have $\frac{A(L)}{c\OPT(L)}\leq 1+\eps/c$.
When the first piece $p$ of a stream $L$ is given, $A'$ chooses a scale factor $\lambda>0$ big enough that when $p$ is scaled up by $\lambda$, the resulting piece $p'\mydef \lambda p$ alone has cost $n$ (i.e., the area or the perimeter of $p'$ is $n$).
The algorithm $A'$ now imitates the strategy of $A$ on the stream $\lambda L$ we get by scaling up all pieces of $L$ by $\lambda$.
We then get that
\[
\frac{A'(L)}{(c+\eps)\OPT(L)}=\frac{A(\lambda L)}{(c+\eps)\OPT(\lambda L)}\leq \frac{(1+\eps/c)c}{c+\eps}=1.\qedhere
\]
\end{proof}

For this reason, we do not distinguish between absolute and asymptotic competitive ratios when the ratio is a constant.
Note that the argument does not work when the competitive ratio is a non-constant function of $\OPT$.
\else
However, for the problems studied in this paper, any constant asymptotic competitive ratio can be matched to within an arbitrarily small difference by an absolute competitive ratio; see the full version for the details.
\fi

\pparagraph{Results and structure of the paper}
We develop online algorithms for the perimeter versions \perirot\ and \peritrans, both with a competitive ratio slightly less than $\compratrot$. % and $\comprattrans$, respectively.
% In the special case that all the given rectangles are squares, we get $\compratpersq$.
These algorithms are described in Section~\ref{sec:perimeter}.
The idea is to partition the positive quadrant into \emph{bricks}, which are axis-parallel rectangles with aspect ratio $\sqrt 2$.
In each brick, we build a stack of pieces which would be too large to place in a brick of smaller size.
Online packing algorithms using higher-dimensional bricks were described by Januszewski and Lassak~\cite{januszewski1997line} and our algorithms are inspired by an algorithm of Fekete and Hoffmann~\cite{DBLP:journals/algorithmica/FeketeH17} that we will get back to.
Interestingly, we show in Section~\ref{sec:alg:inferior} that a more direct adaptation of the algorithm of Fekete and Hoffmann has a competitive ratio of at least $4$, and is thus inferior to the algorithm we describe.
We also give a lower bound of $4/3$ for the version with translations and $5/4$ for the version with rotations.

In Section~\ref{sec:area}, we study the area versions \arearot\ and \areatrans.
We show in Section~\ref{sec:lower:general} that for any algorithm $A$ processing a stream of $n$ pieces cannot achieve a better competitive ratio than $\Omega(\sqrt{n})$, and this holds for all online algorithms and with and without rotations allowed.
It also holds in the special case where all the edges of pieces have length at least $1$.
We furthermore show that when the pieces can be arbitrary, there can be given no bound on the competitive ratio as a function of $\OPT$ for $\arearot$ nor $\areatrans$.
In Section~\ref{sec:alg:arbitrary} we describe the algorithms $\dynamicboxtrans$ and $\dynamicboxrot$, which achieve a $O(\sqrt{n})$ competitive ratio for $\areatrans$ and $\arearot$, respectively, for an arbitrary stream of $n$ pieces.
This is thus optimal up to a constant factor when measuring the competitive ratio as a function of $n$.
Both algorithms use a row of boxes of exponentially increasing width and dynamically adjusted height. In these boxes, we pack pieces using a next-fit shelf algorithm, which is a classic online strip packing algorithm first described by Baker and Schwartz~\cite{baker1983shelf}.

% The $\Omega(\sqrt{n})$ lower bound relies on feeding the algorithm with rectangles with arbitrarily large aspect ratios (the ratio between the longest and the shortest edge of a rectangle) and with arbitrarily short edges.
We then turn our attention to two special cases. %, each excluding one of these extremes.

The first special case is when the aspect ratio is bounded by a constant $\alpha\geq 1$.
A case of particular interest is when all pieces are squares, i.e., $\alpha=1$.
It is natural to have the same requirement to the container as to the pieces, so let us assume that the goal is to minimize the area of the axis-parallel bounding square of the pieces, and call the problem \sqArea.
This problem was studied by Fekete and Hoffmann~\cite{DBLP:journals/algorithmica/FeketeH17}, and they gave an algorithm for the problem and proved that it was $8$-competitive.
We prove that the same algorithm is in fact $6$-competitive and that this is tight.
It easily follows that if the aspect ratio is bounded by an arbitrary constant $\alpha\geq 1$ or if the goal is to minimize the area of the axis-parallel bounding rectangle, we also get a $O(1)$-competitive algorithm.

The second special case is when all edges are \emph{long}, that is, when they have length at least $1$ (any other constant will work too).
In Section~\ref{sec:lower:restricted}, we show that under this assumption, there is a lower bound of $\Omega(\sqrt{\OPT})$ for the asymptotic competitive ratio of $\areatrans$, whereas for $\arearot$, we get the lower bound $\Omega(\sqrt[4]{\OPT})$.
In Section~\ref{alg:area}, we provide algorithms for the area versions when the edges are long.
For both problems \arearot\ and \areatrans, we give algorithms that match the lower bounds of Section~\ref{sec:lower:restricted} to within a constant factor.
With translations only, this is just the algorithm from the general case with competitive ratio $O(\sqrt n)=O(\sqrt{\OPT})$.
The algorithm with ratio $O(\sqrt[4]{\OPT})$ for the rotational case follows the same scheme as the algorithms for arbitrary rectangles of Section~\ref{sec:alg:arbitrary}, but differ in the way we dynamically increase boxes' heights.
We finally describe an algorithm for the rotational case with competitive ratio $O(\min\{\sqrt n,\sqrt[4]{\OPT}\})$, thus matching the lower bounds $\Omega(\sqrt n)$ and $\Omega(\sqrt[4]{\OPT})$ simultaneously. Actually, the two lower bounds for $\arearot$ can be summarized by $\Omega(\max\{\sqrt{n}, \sqrt[4]{\OPT}\})$, while we manage to achieve a competitive ratio of $O(\min\{\sqrt n,\sqrt[4]{\OPT}\})$. However, this gives no contradiction, it simply proves that the \emph{edge cases} that have a competitive ratio of at least $\Omega(\sqrt[4]{\OPT})$ must satisfy $\OPT = O(n^2)$, and those for which the competitive ratio is at least $\Omega(\sqrt{n})$ satisfy $n = O(\sqrt{\OPT})$.

We summarize the results in Table~\ref{table:res}.

\begin{table}[]
\centering
\begin{tabular}{|l|l|l|l|l|}
\hline
\textbf{Measure}        & \textbf{Version}         & \textbf{Trans./Rot.} & \textbf{Lower bound} & \textbf{Upper bound}   \\ \hline
\multirow{2}{*}{Perimeter} & \multirow{2}{*}{General}    & Translation          & $4/3$,~Sec.~\ref{sec:lowerB}   & $4-\eps$,~Sec.~\ref{sec:algorithms}  \\ \cline{3-5} 
   & & Rotation & $5/4$,~Sec.~\ref{sec:lowerB}   & $4-\eps$,~Sec.~\ref{sec:algorithms}  \\ \hline
\multirow{9}{*}{Area}      & \multirow{3}{*}{General}    & Translation          & \thead{$\Omega(\sqrt n)$ \& $\forall f: \Omega(f(\OPT))$, \\ Sec.~\ref{sec:lower:general}} & $O(\sqrt n)$, Sec.~\ref{sec:alg:arbitrary}             \\ \cline{3-5} 
   & & Rotation & \thead{$\Omega(\sqrt n)$ \& $\forall f: \Omega(f(\OPT))$, \\ Sec.~\ref{sec:lower:general}} & $O(\sqrt n)$, Sec.~\ref{sec:alg:arbitrary}            \\ \cline{2-5} 
   & Sq.-in-sq. & N/A & 16/9, Sec.~\ref{sec:boundedaspect} & $6$, Sec.~\ref{sec:boundedaspect} \\ \cline{2-5} 
   & \multirow{3}{*}{Long edges} & Translation          & $\Omega(\sqrt{\OPT})$, Sec.~\ref{sec:lower:restricted}         & \thead{$O(\sqrt n)=O(\sqrt{\OPT})$, \\ Sec.~\ref{alg:area}}         \\ \cline{3-5} 
   & & Rotation & \thead{$\Omega(\max\{\sqrt n, \sqrt[4]{\OPT}\})$, \\ Sec.~\ref{sec:lower:general} and~\ref{sec:lower:restricted}}  & \thead{$O(\min\{\sqrt{n},\sqrt[4]{\OPT}\})$, \\ Sec.~\ref{alg:area}} \\ \hline
\end{tabular}
\caption{Results of this paper.}
\label{table:res}
\end{table}

\pparagraph{Related work}
The literature on online packing problems is rich.
See the surveys of Christensen, Khan, Pokutta, and Tetali~\cite{CHRISTENSEN201763}, van Stee~\cite{DBLP:journals/sigact/Stee12,DBLP:journals/sigact/Stee15}, and Csirik and Woeginger~\cite{Csirik1998} for an overview.
It seems that the vast majority of previous work on online versions of two-dimensional packing problems is concerned with either bin packing (packing the pieces into a minimum number of unit squares) or strip packing (packing the pieces into a strip of unit width so as to minimize the total height of the pieces).
From a mathematical point of view, we find the problems studied in this paper perhaps even more fundamental than these important problems in the sense that we give no restrictions on where to place the pieces, whereas the pieces are restricted by the boundaries of the bins and the strip in bin and strip packing.
%However, there seems to be no prior work on the online versions of the packing problems studied in this paper.

Another related problem is to find the critical density of online packing squares into a square.
In other words, what is the maximum $\Sigma\leq 1$ such that there is an online algorithm that packs any stream of squares of total area at most $\Sigma$ into the unit square?
This was studied, among others, by Fekete and Hoffmann~\cite{DBLP:journals/algorithmica/FeketeH17} and Brubach~\cite{Brubach14}.
Lassak~\cite{lassak1997linepot} and Januszewski and Lassak~\cite{januszewski1997line} studied higher-dimensional versions of this problem.

\iffull
Milenkovich~\cite{milenkovic1996translational} studied generalized offline versions of the minimum area problem:
Translate $k$ given $m$-gons into a convex container of minimum area with edges in $n$ fixed directions.
When the $m$-gons can be non-convex, the running time is $O((m^2+n)^{2k-2}(n+\log m))$, and when they are convex, the running times are $O((m+n)^{2k}(n+\log m))$ or $O(m^{k-1}(n^{2k+1}+\log m))$.
Milenkovich and Daniels~\cite{doi:10.1111/j.1475-3995.1999.tb00171.x} described different algorithms for the same problems.
Milenkovich~\cite{MILENKOVIC19993} also studied the same problem when arbitrary rotations are allowed and the container is either a strip with a fixed width, a homothet of a given convex polygon, or an arbitrary rectangle (as in our work).
He gave $(1+\eps)$-approximation algorithms (no explicit running times are given, but they are apparently also exponential).
\else
Milenkovich~\cite{milenkovic1996translational,MILENKOVIC19993} and Milenkovich and Daniels~\cite{doi:10.1111/j.1475-3995.1999.tb00171.x} studied generalized offline versions of the minimum area problem where the pieces are simple or convex polygons.
\fi

Some algorithms have been described for computing the packing of two or three convex polygons that minimizes the perimeter or area of the convex hull or the bounding box~\cite{ahn2012aligning, althurtado,leewoo,PARK20161}. %,GRINDE1997231}.

\iffull
Alt~\cite{DBLP:journals/eatcs/Alt16} demonstrated how a $\rho$-approximation algorithm for strip packing (axis-parallel rectangles with translations) can be turned into a $(1+\eps)\rho$-approximation algorithm for the offline version of \areatrans, for any constant $\eps>0$.
The same technique works for $\arearot$.
The idea is to apply the strip packing algorithm to strips of increasing widths and in the end choose the packing that resulted in the smallest area.
Therefore, the same technique cannot be applied in the online setting, where we need to choose a placement for each piece and stick with it.
Alt also mentioned that finding a minimum area bounding box of a set of convex polygons with arbitrary rotations allowed can be reduced to the problem where the pieces are axis-parallel rectangles with only translations allowed.
This reduction increases the approximation ratio by a factor by $2$.
The reduction does not work when the pieces can be only translated, but Alt, de Berg, and Knauer~\cite{altapproximating} gave a $17.45$-approximation algorithm for this problem using different techniques.
\else
Alt~\cite{DBLP:journals/eatcs/Alt16} and Alt, de Berg, and Knauer~\cite{altapproximating} gave constant factor approximation algorithms for the offline versions of \areatrans\ and \arearot\ when the pieces are axis parallel rectangles or convex polygons, with translations only or arbitrary rotations allowed.
\fi

Lubachevsky and Graham~\cite{lubachevsky2003dense} used computational experiments to find the rectangles of minimum area into which a given number $n\leq 5000$ of congruent circles can be packed; see also the follow-up work by Specht~\cite{SPECHT201358}.
In another paper, Lubachevsky and Graham~\cite{LUBACHEVSKY20091947} studied the problem of minimizing the perimeter instead of the area.

Another fundamental packing problem is to find the smallest square containing a given number of \emph{unit} squares, with arbitrary rotations allowed. 
A long line of mathematical research has been devoted to this problem, initiated by Erd\H{o}s and Graham~\cite{erdos1975packing} in 1975, and it is still an active research area~\cite{chung2019efficient}.

% As it appears, algorithmic problems where the goal is to minimize the area have received more attention than those of minimizing the perimeter, perhaps because the area is a direct measure of the cost when the pieces are to be cut out from a large sheet of material.
% However, minimizing the perimeter also arises naturally in practical contexts, for instance when the pieces have to be enclosed by a fence of minimum length.

\section{The perimeter versions}\label{sec:perimeter}
In Section~\ref{sec:algorithms}, we present two online algorithms to minimize the perimeter of the bounding box:
the algorithm $\smallboxtrans$ solves the problem $\peritrans$, where we can only translate pieces; the algorithm $\smallboxrot$ solves the problem \perirot, where also rotations are allowed. Both algorithms achieve a competitive ratio of $\compratrot$.
In Section~\ref{sec:lowerB}, we show a lower bound of $4/3$ for the version with translations and $5/4$ for the version with rotations.

\subsection{Algorithms to minimize perimeter}\label{sec:algorithms}
\pparagraph{Algorithm for translations}
We pack the pieces into non-overlapping \emph{bricks}; a technique first described by Januszewski and Lassak~\cite{januszewski1997line} which was also used by Fekete and Hoffmann~\cite{DBLP:journals/algorithmica/FeketeH17} for the problem \sqArea.
Let a \emph{$k$-brick} be a rectangle of size $\sqrt 2^{-k}\times \sqrt 2^{-k-1}$ if $k$ is even and $\sqrt 2^{-k-1}\times \sqrt 2^{-k}$ if $k$ is odd.
A \emph{brick} is a $k$-brick for some integer $k$.
% rectangle with aspect ratio $\sqrt 2$, i.e., one pair of edges are $\sqrt 2$ times longer than the other edges.

We tile the positive quadrant using one $k$-brick $B_k$ for each integer $k$ as in Figure~\ref{fig:brickPack} (left):
if $k$ is even, $B_k$ is the $k$-brick with lower left corner $(0,\sqrt 2^{-k-1})$ and otherwise, $B_k$ is the $k$-brick with lower left corner $(\sqrt 2^{-k-1},0)$.  The bricks $B_k$ are called the \emph{fundamental} bricks.
We define $B_{>k}\mydef\bigcup_{i>k} B_i$ and $B_{\geq k}\mydef B_{>k-1}$, so that $B_{>k}$ is the $k$-brick immediately below (if $k$ is even) or to the left (if $k$ is odd) of $B_k$.

An important property of a $k$-brick $B$ is that it can be split into two $(k+1)$-bricks: $B\spl 1$ and $B\spl 2$; see Figure~\ref{fig:brickPack} (middle). We introduce a uniform naming and define $B\spl 1$ to be the left half of $B$ if $k$ is even and the lower half of $B$ if $k$ is odd.

%\mikkel{Should we introduce shorthand notation $B\spl b_1b_2\ldots b_k=B\spl b_1\spl b_2\spl \ldots \spl b_k$?}
%\lorenzo{After trying it, I prefer the multiple dagger notation}

We define a \emph{derived} brick recursively as follows: a derived brick is either (i) a fundamental brick $B_k$ or (ii) $B\spl 1$ or $B\spl 2$, where $B$ is a derived brick.
We introduce an ordering $\prec$ of the derived $k$-bricks as follows.
Consider two derived $k$-bricks $D_1$ and $D_2$ such that $D_1\subset B_i$ and $D_2\subset B_j$. If $i>j$, then $D_1 \prec D_2$. Else, if $i=j$ then the bricks $D_1$ and $D_2$ are both obtained by splitting the fundamental brick $B_i$, and the number of splits is $\ell\mydef i-k$.  Hence the bricks have the forms $D_1=B_i\spl b_{11}\spl b_{12}\spl \ldots\spl b_{1\ell}$ and $D_2=B_i\spl b_{21}\, b_{22}\, \ldots\, b_{2\ell}$, where $b_{ij}\in\{1,2\}$ for $i\in\{1,2\}$ and $j\in\{1,\ldots,\ell\}$.
We then define $D_1 \prec D_2$ if $(b_{11},b_{12},\ldots,b_{1\ell})$ precedes $(b_{21},b_{22},\ldots,b_{2\ell})$ in the lexicographic ordering.

We say that a $k$-brick is \emph{suitable} for a piece $p$ of size $w\times h$ if the width and height of the brick are at least $w$ and $h$, respectively, and if that is not the case for a $(k+1)$-brick.  We will always pack a given piece $p$ in a derived $k$-brick that is suitable for $p$.

\begin{figure}
\centering
\includegraphics[page=4]{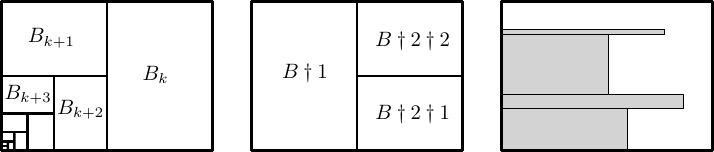}
\caption{Left: Fundamental bricks. Middle: Splitting a brick. 
Right: Rectangular pieces packed in a brick.}
\label{fig:brickPack}
\end{figure}

We now explain how we pack pieces into one specific brick; see Figure~\ref{fig:brickPack} (right).
The first piece $p$ that is packed in a brick $B$ is placed with the lower left corner of $p$ at the lower left corner of $B$.  Suppose now that some other pieces $p_1,\ldots,p_i$ have been packed in $B$.
If $k$ is even, then $p_1,\ldots,p_i$ form a stack with the left edges contained in the left edge of $B$, and we place $p$ on top of $p_i$ (again, with the left edge of $p$ contained in the left edge of $B$).
Otherwise, $p_1,\ldots,p_i$ form a stack with the bottom edges contained in the bottom edge of $B$, and we place $p$ to the right of $p_i$ (again, with the bottom edge of $p$ contained in the bottom edge of $B$). We say that a brick \emph{has room} for a piece $p$ if the packing scheme above places $p$ within $B$, and it is apparent that an empty suitable brick for $p$ has room for $p$.

\begin{figure}
\centering
\includegraphics[page=5]{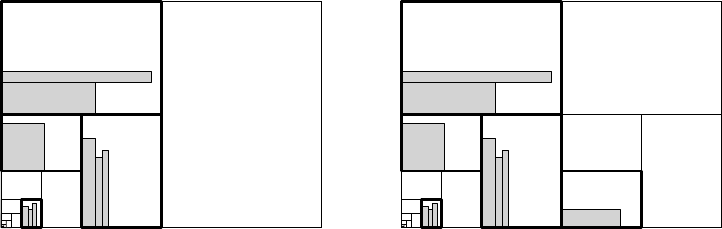}
\caption{Left: Some pieces have been packed by the algorithm.
The bricks in $\mathcal D$ are drawn with fat edges.
Right: A new piece arrives.
There is already a brick of the suitable size in $\mathcal D$, but there is not enough room, so a new brick of the same size is added to $\mathcal D$ where the piece is placed.}
\label{fig:brickPackAlg}
\end{figure}

The algorithm \smallboxtrans\ maintains the collection $\mathcal D$ of non-overlapping derived bricks, such that one or more pieces have been placed in each brick in $\mathcal D$; see Figure~\ref{fig:brickPackAlg}.
% Each rectangle is packed in one of the bricks in $\mathcal D$. %, and if there is not enough space in any brick in $\mathcal D$, we add copies of larger fundamental bricks to $\mathcal D$.
Before the first piece arrives, we set $\mathcal D\mydef \emptyset$.
Suppose that some stream of pieces have been packed, and that a new piece $p$ appears.
Choose $k$ such that a $k$-brick is suitable for $p$. If there exists a derived $k$-brick $D\in\mathcal D$ such that $D$ has room for $p$, then we pack $p$ in $D$. Else, let $D$ be the minimum derived $k$-brick (with respect to the ordering $\prec$ described before) such that $D$ is interior-disjoint from each brick in $\mathcal D$; we then add $D$ to $\mathcal D$ and pack $p$ in $D$.

\begin{theorem}\label{thm:peritransalg}
The algorithm $\smallboxtrans$ has a competitive ratio strictly less than \comprattrans\ for \peritrans.
\end{theorem}

\begin{proof}
\iffull
We can assume, without loss of generality, that after we have packed the last rectangle, we have $\bigcup \D \subseteq B_{\geq0}$ and $\bigcup \D \not\subseteq B_{\geq1}$.
% To obtain this property, it is sufficient to rescale the whole configuration $\mathcal D$ of a factor $\sqrt{2}^j$ for some $j \in \Z$. \mikkel{Is the following true?} If $j$ is odd, we furthermore flip the arrangement along the line $y=x$. The resulting configuration $\D^\prime$ is exactly the configuration we would obtain running $\smallboxtrans$ over a modified stream obtained from the original one by rescaling each piece by the same factor $\sqrt{2}^j$ (and rotating it by $90^\circ$ if $j$ is odd).  This can be proven by a straightforward induction. Finally, the offline optimal perimeter $\OPT$ is covariant with respect to linear transformations, therefore the competitive ratio is preserved by the process.
As shown in Figure~\ref{fig:sparse_dense}, we define a derived $k$-brick $B \subseteq B_{\geq0}$ to be
\begin{itemize}
    \item \emph{sparse} if $B \in \D$ and the total height (if $k$ is even, else width) of pieces stacked in $B$ is less than half of the height (if $k$ is even, else width) of $B$,
    \item \emph{dense} if $B \in \D$ and $B$ it is not sparse,
    \item \emph{free} if $B$ is interior-disjoint from each brick in $\D$, and
    \item \emph{empty} if $B$ is a maximal (w.r.t.~inclusion) free brick.
\end{itemize}
\begin{figure}
    \centering
    \includegraphics[page=10]{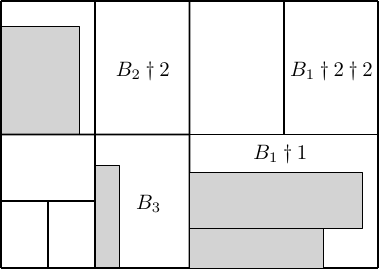}
    \caption{Brick $B_3$ is \emph{sparse}, brick $B_2 \spl 2$ is \emph{empty}, brick $B_1 \spl 1$ is \emph{dense}. Brick $B_1 \spl 2 \spl 2$ is free, but not empty, since it is contained in $B_1 \spl 2$.}
    \label{fig:sparse_dense}
\end{figure}

\begin{remark} \label{remark:cover}
Sparse, dense and empty bricks together cover $B_{\geq 0}$, in fact every brick in $\D$ is either sparse or dense, and any brick in $B_{\geq 0}$ that is interior-disjoint from bricks in $\D$ is contained in some empty brick. 
\end{remark}

\begin{remark} \label{remark:edge_length}
Every $k$-brick $D \in \D$ contains pieces for which it is suitable. Therefore, if $k$ is odd $D$ contains a piece of height at least $\sqrt{2}^{-k} / 2$, and if $k$ is even $D$ contains a piece of width at least $\sqrt{2}^{-k} / 2$.  
\end{remark}

\begin{remark} \label{remark:density}
Every $k$-brick $D \in \D$ that is dense contains pieces with total area at least $1 / 4$ of the area of $D$. To see this, suppose that $k$ is even, so that $D$ is $\sqrt{2}^{-k} \times \sqrt{2}^{-k-1}$, then thanks to density the total height of pieces in $D$ is at least half of its height, moreover thanks to Remark~\ref{remark:edge_length} all the pieces contained in $D$ have width at least $\sqrt{2}^{-k} / 2$. If $k$ is odd we prove it analogously.  
\end{remark}

\begin{remark} \label{remark:free_monotonicity}
Consider two $k$-bricks $M$ and $N$.
If $M \prec N$ and $M$ is free, then $N$ is free. To prove this it is sufficient to consider the step in which the first piece $p$ is placed within $N$ and $N \spl b_{1} \spl \ldots \spl b_\ell$ is added to $\D$. Then, $N \spl b_{1} \spl \ldots \spl b_\ell$ should be the $\prec$-minimum free suitable $k$-brick, but $M \spl b_{1} \spl \ldots \spl b_\ell \prec N \spl b_{1} \spl \ldots \spl b_\ell$ gives a contradiction.
It follows that whenever we have a set $S$ of $k$-bricks that contains a free $k$-brick, then also $\max_\prec S$ is free. This turns out to be useful multiple times along the proof, choosing $S$ to be the set of $k$-bricks not contained in a strictly larger empty brick.
\end{remark}

\begin{remark} \label{remark:unicity_sparse_empty}
There exists no empty $0$-brick, otherwise $D \subseteq B_{\geq 1}$. Moreover, for every $k \geq 1$ we can have at most one sparse $k$-brick and one empty $k$-brick. In fact, a new empty (resp.~sparse) $k$-brick is created only when no empty (resp.~sparse) $k$-brick exists.
\end{remark}

In the following we prove an upper bound on the competitive ratio $\ALG / \OPT$, where $\ALG$ is the perimeter of the bounding box achieved by our online algorithm and $\OPT$ is the optimal perimeter computed offline.  Hence, we need some techniques to provide an upper bound on $\ALG$ and a lower bound on $\OPT$. For $\ALG$, we will simply show a bounding box, in fact the perimeter of any bounding box containing all the pieces provides an upper bound to the minimum perimeter bounding box. For $\OPT$, let $A$ be the total area of pieces and $L$ be the maximum length of an edge of a piece. If $L^2 > A$ then the minimum perimeter bounding box cannot have a smaller perimeter than a box of size $L \times A / L$. Otherwise, if $L^2 \leq A$ we have a weaker lower bound given by the box $\sqrt{A} \times \sqrt{A}$. 
%To rigorously obtain the previous bounds it is sufficient to minimize the perimeter under the constraints of having area at least $A$ and an edge of length at least $L$.
%\mikkel{Is the previous sentence needed?}
Throughout the analysis we consider semiperimeters instead of perimeters to improve readability.

We denote with $A(empty), A(sparse), A(dense)$ the total area of empty, sparse and dense bricks respectively. Thanks to Remark~\ref{remark:cover}, we have that $A(empty) + A(sparse) + A(dense) = A(B_{\leq 0}) = \sqrt{2}$. We denote with $A_{pcs}$ the total area of pieces in the stream. Thanks to Remark~\ref{remark:density}, we have $A_{pcs} \geq A(dense) / 4$. From now on the proof branches in many cases and subcases.
We will perform a depth-first visit of the case tree, and for each leaf of this tree we will prove that the competitive ratio is strictly less than $\comprattrans$. 
Let $k$ be the smallest integer such that there exists a $k$-brick in $\D$, and let $M \in \D$ be the $\prec$-maximal $k$-brick.
From our assumptions, it follows that $k\geq 0$.
\paragraph*{Case Tree}

\begin{figure}
\centering
\includegraphics[scale=0.75]{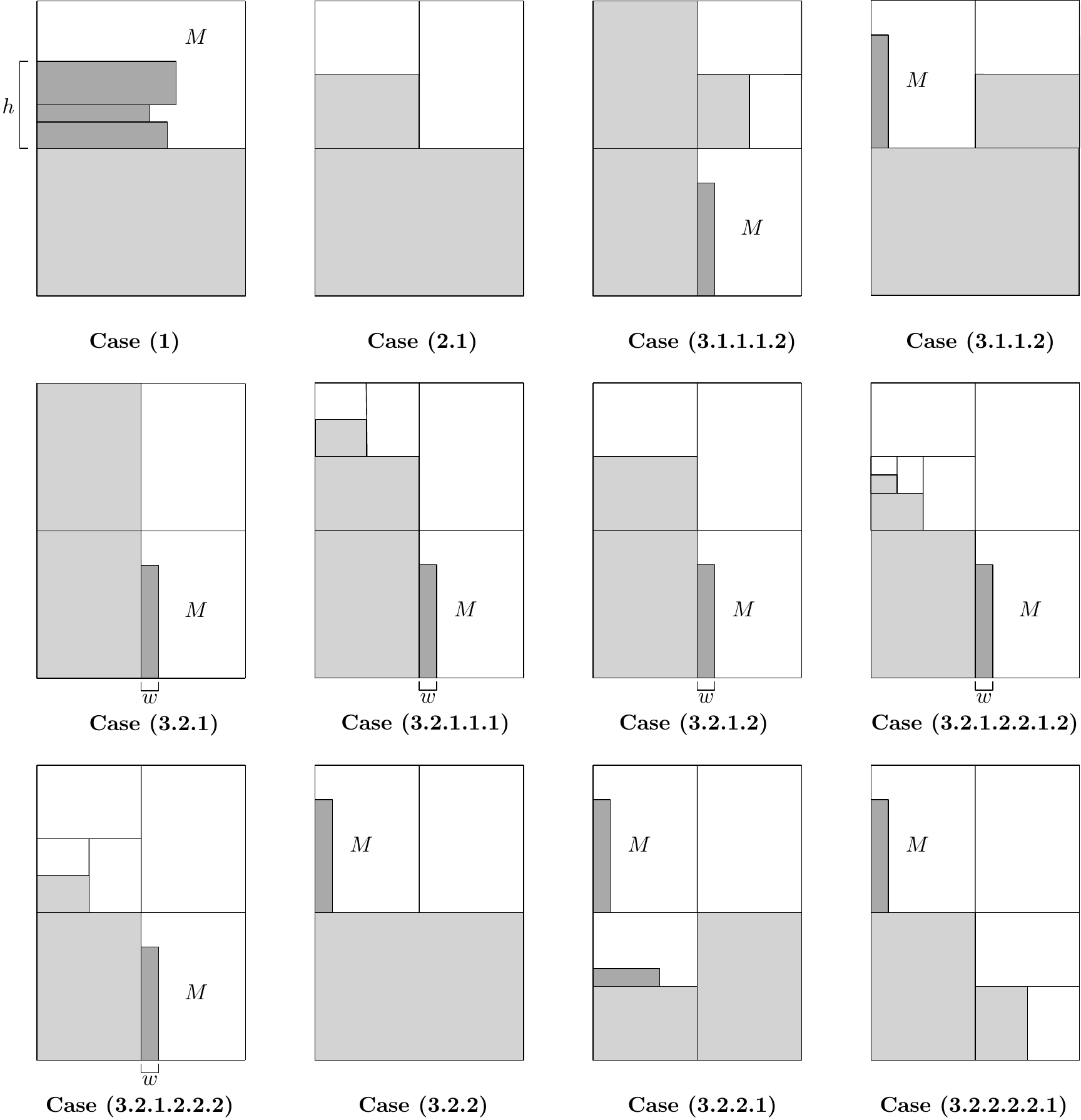}
\caption{Some of the cases listed in the proof of Theorem~\ref{thm:peritransalg} are shown. The grey area must fit within the bounding box considered in the case analysis.}
\label{fig:case_tree}
\end{figure}

\case{1}{$M$ is a $0$-brick} 

Thanks to Remark~\ref{remark:edge_length} every piece in $M$ has width at least $1/2$. Let $h$ be the total height of pieces stacked in $M$, then a bounding box of size size $1 \times (\sqrt{2} / 2 + h)$ is obtained cutting the topmost part of $M$; see Figure~\ref{fig:case_tree}. We can easily bound $\OPT$ with $1 / 2 \times h$, and we get
\[
\frac{\ALG}{\OPT} \leq \frac{1 + \frac{\sqrt{2}}{2} + h}{\frac{1}{2} + h} \leq  2 + \sqrt{2} < \comprattrans.
\]
\case{2}{$M$ is a $k$-brick for $k \geq 2$} 

Here we have two cases.

\case{2.1}{There exist a $1$-brick $N_1$ and a $2$-brick $N_2$ that are empty} 

Thanks to Remark~\ref{remark:free_monotonicity} we can choose $N_1$ = $B_0 \spl 2$ and $N_2 = B_0 \spl 1 \spl 2$. In fact, for $B_0 \spl 2$ it is sufficient to choose $S$ as the set of all $1$-bricks, while for $B_0 \spl 1 \spl 2$ we can choose $S$ to be the set of all $2$-bricks that are not contained in a larger free brick. Thus, we can cut the topmost half of $B_0$ and get $\ALG \leq 1 + 3 / 4 \cdot \sqrt{2}$; see Figure~\ref{fig:case_tree}. We have 
\begin{align*}
& A(empty) \leq \sum_{i \geq 1} A(B_i) \leq \frac{\sqrt{2}}{2} \\
& A(sparse) \leq \sum_{i \geq 2} A(B_i) = \frac{\sqrt{2}}{4} \quad \text{(thanks to case (2) clause there is no sparse $1$-brick)} \\
& A_{pcs} \geq \frac{A(dense)}{4} \geq \frac{A(B_{\geq 0}) - A(sparse) - A(empty)}{4} \geq \frac{\sqrt{2}}{16}
\end{align*}
Now we are ready to bound $\OPT$:
\begin{align*}
& \OPT \geq 2 \cdot \sqrt{A_{pcs}} = \sqrt{\frac{\sqrt{2}}{4}} \\
& \frac{\ALG}{\OPT} \leq \frac{1 + \frac{3}{4} \sqrt{2}}{\sqrt{\frac{\sqrt{2}}{4}}} \approx 3.47 < 4.
\end{align*}

\case{2.2}{For $j = 1$ or $j = 2$ there does not exist an empty $j$-brick} 

In this case we just use $\ALG \leq 1 + \sqrt{2}$.  Then we have 
\begin{align*}
& A(empty) \leq \sum_{i \geq 1 \land i \neq j} A(B_i) \leq \frac{3}{8} \sqrt{2} \quad \text{(worst case is when $j = 2$)}
\\
& A(sparse) \leq \sum_{i \geq 2} A(B_i) = \frac{\sqrt{2}}{4}
\end{align*}
therefore performing the same computations of case (2.1), $A_{pcs} \geq 3 / 32 \cdot \sqrt{2}$, and finally
\begin{align*}
& \OPT \geq 2 \cdot \sqrt{\frac{3}{32} \sqrt{2}} = \sqrt{\frac{3}{8} \sqrt{2}} \\ 
& \frac{\ALG}{\OPT} \leq \frac{1 + \sqrt{2}}{\sqrt{\frac{3}{8} \sqrt{2}}} \approx 3.32 < 4.
\end{align*}

\case{3}{$M$ is a $1$-brick} 

For the rest of the proof $L$ will be the length of the longest edge among all pieces. Since $M$ is a $1$-brick, we have $\sqrt{2} / 4 < L \leq \sqrt{2} / 2$. Here we have two cases.

\case{3.1}{There does not exists an empty $1$-brick}

Here we have two cases.

\case{3.1.1}{For $j = 2$ and $j = 3$ there exists an empty $j$-brick}

Here we have three cases.

\case{3.1.1.1}{$M$ is the fundamental brick $B_1$}

Thanks to Remark~\ref{remark:free_monotonicity} we can assume $B_0 \spl 2 \spl 2$ and $B_0 \spl 2\spl 1 \spl 2$ to be empty. Here we have two cases.

\case{3.1.1.1.1}{$M$ is dense}

Since $M=B_1$ is the $\prec$-maximal $k$-brick in $\D$, then there does not exist a sparse $1$-brick.
\begin{align*}
& A(empty) \leq \sum_{i \geq 2} A(B_i) \leq \frac{\sqrt{2}}{4}
\\
& A(sparse) \leq \sum_{i \geq 2} A(B_i) = \frac{\sqrt{2}}{4}
\end{align*}
therefore $A_{pcs} \geq \frac {\sqrt{2}} 8$, and finally
\begin{align*}
& \OPT \geq 2 \cdot \sqrt{\frac {\sqrt{2}} 8} = \sqrt{\frac {\sqrt{2}} 2} \\ 
& \frac{\ALG}{\OPT} \leq \frac{1 + \sqrt{2}}{\sqrt{\frac {\sqrt{2}} 2}} \approx 2.87 < 4.
\end{align*}

\case{3.1.1.1.2}{$M$ is sparse}

Then, we can cut the rightmost part of $B_{\geq 0}$ and get a $3/4 \times \sqrt{2}$ bounding box; see Figure~\ref{fig:case_tree}. We have 
\begin{align*}
& A(empty) \leq \sum_{i \geq 2} A(B_i) \leq \frac{\sqrt{2}}{4}\\
& A(sparse) \leq \sum_{i \geq 1} A(B_i) \leq \frac{\sqrt{2}}{2}
\end{align*}
hence $A_{pcs} \geq \sqrt{2} / 16$. Since $L^2 > 1 / 8 > \sqrt{2} / 16$ we finally have 
\begin{align*}
& \OPT \geq L + \frac{A_{pcs}}{L} \geq \frac{\sqrt{2}}{4} + \frac{1}{4} \quad \text{(minimizing over $L \in [\sqrt{2} / 4, \sqrt{2} / 2]$)}\\
& \frac{\ALG}{\OPT} \leq \frac{3/4 + \sqrt{2}}{\frac{\sqrt{2}}{4} + \frac{1}{4}} \approx 3.59 < 4.
\end{align*}

\case{3.1.1.2}{$M = B_0 \spl 1$}

Thanks to Remark~\ref{remark:free_monotonicity} we can assume $B_0 \spl 2 \spl 2$ to be empty. Then, we can cut the topmost part of $B_{\geq 0}$ and get a $1 \times \sqrt{2}/2 + L$ bounding box; see Figure~\ref{fig:case_tree}. We have 
\begin{align*}
& A(empty) \leq \sum_{i \geq 2} A(B_i) \leq \frac{\sqrt{2}}{4}\\
& A(sparse) \leq \sum_{i \geq 1} A(B_i) \leq \frac{\sqrt{2}}{2}
\end{align*}
hence $A_{pcs} \geq \sqrt{2} / 16$. Since $L^2 > 1 / 8 > \sqrt{2} / 16$ we finally have 
\begin{align*}
& \OPT \geq L + \frac{A_{pcs}}{L} \geq L + \frac{\sqrt{2}}{16 L} \\
& \frac{\ALG}{\OPT} \leq \frac{1 + \sqrt{2}/2 + L}{L + \frac{\sqrt{2}}{16 L}} \leq 2 + \sqrt{2} < 4. \quad \text{(maximizing over $L \in [\sqrt{2} / 4, \sqrt{2} / 2]$)}
\end{align*}

\case{3.1.1.3}{$M = B_0 \spl 2$} 

This case is analogous to the previous one, in fact thanks to Remark~\ref{remark:free_monotonicity} we can assume $B_0 \spl 1 \spl 2$ to be empty and cut the topmost part of $B_{\geq 0}$.

\case{3.1.2}{For $j = 2$ or $j = 3$ there does not exist an empty $j$-brick}

\begin{align*}
& A(empty) \leq \sum_{i \geq 2 \land i \neq j} A(B_i) \leq \frac{3}{16} \sqrt{2} \quad \text{(worst case is when $j = 3$)}\\
& A(sparse) \leq \sum_{i \geq 1} A(B_i) \leq \frac{\sqrt{2}}{2}
\end{align*}
hence $A_{pcs} \geq \frac 5 {64} \cdot \sqrt{2} $. Since $L^2 > 1 / 8 > \frac 5 {64} \cdot \sqrt{2}$ we finally have 
\begin{align*}
& \OPT \geq L + \frac{A_{pcs}}{L} \geq \frac{\sqrt{2}}{4} + \frac{5}{16} \quad \text{(minimizing over $L \in [\sqrt{2} / 4, \sqrt{2} / 2]$)}\\
& \frac{\ALG}{\OPT} \leq \frac{1 + \sqrt{2}}{\frac{\sqrt{2}}{4} + \frac{5}{16}} \approx 3.62 < 4.
\end{align*}

\case{3.2}{There exists an empty $1$-brick} 

Thanks to Remark~\ref{remark:free_monotonicity} we can assume $B_0 \spl 2$ to be empty. Here we have two cases.

\case{3.2.1}{$M$ is the fundamental brick $B_1$} 

Let $w$ be the total width of pieces stacked in $M$. Since $B_0 \spl 2$ is empty, we can cut the rightmost part of $B_{\geq 0}$ and get a $(1 / 2 + w) \times \sqrt{2}$ bounding box; see Figure~\ref{fig:case_tree}. Since increasing $w$ only improves our estimates, we consider the corner case $w = 0$. Now we have two cases.

\case{3.2.1.1}{There does not exist an empty $2$-brick}

Here we have two cases.

\case{3.2.1.1.1}{For $j = 3$ and $j = 4$ there exists an empty $j$-brick}

Thanks to Remark~\ref{remark:free_monotonicity} we can assume $B_0 \spl 1 \spl 2 \spl 2$ and $B_0 \spl 1\spl 2 \spl 1 \spl 2$ to be empty. Thus, we can cut the topmost part of $B_{\geq 0}$ and get a $1/2 \times (7/8 \cdot \sqrt{2})$ bounding box; see Figure~\ref{fig:case_tree}. We have

\begin{align*}
& A(empty) \leq \sum_{i \geq 1 \land i \neq 2} A(B_i) \leq \frac{3}{8} \sqrt{2} \\
& A(sparse) \leq \sum_{i \geq 1} A(B_i) \leq \frac{\sqrt{2}}{2}
\end{align*}
hence $A_{pcs} \geq \sqrt{2} / 32$. Since $L^2 > 1 / 8 > \sqrt{2} / 32$ we finally have 
\begin{align*}
& \OPT \geq L + \frac{A_{pcs}}{L} \geq \frac{\sqrt{2}}{4} + \frac{1}{8} \\
& \frac{\ALG}{\OPT} \leq \frac{1/2 + (7/8 \cdot \sqrt{2})}{\frac{\sqrt{2}}{4} + \frac{1}{8}} \approx 3.63 < 4.
\end{align*}

\case{3.2.1.1.2}{For $j = 3$ or $j = 4$ there does not exist an empty $j$-brick}

\begin{align*}
& A(empty) \leq \sum_{i \geq 1 \land i \neq 2, j} A(B_i) \leq \frac{11}{32} \sqrt{2}  \quad \text{(worst case is when $j = 4$)}\\
& A(sparse) \leq \sum_{i \geq 1} A(B_i) \leq \frac{\sqrt{2}}{2}
\end{align*}
hence $A_{pcs} \geq 5/128 \cdot \sqrt{2}$. Since $L^2 > 1 / 8 > 5/128 \cdot \sqrt{2}$ we finally have 
\begin{align*}
& \OPT \geq L + \frac{A_{pcs}}{L} \geq \frac{\sqrt{2}}{4} + \frac{5}{32} \\
& \frac{\ALG}{\OPT} \leq \frac{1/2 + \sqrt{2}}{\frac{\sqrt{2}}{4} + \frac{5}{32}} \approx 3.75 < 4.
\end{align*}

\case{3.2.1.2}{There exists an empty $2$-brick} 

Thanks to Remark~\ref{remark:free_monotonicity} we can assume $B_0 \spl 1 \spl 2$ to be empty. Thus, we can cut the topmost part of $B_{\geq 0}$ and get a $1/2 \times (3 / 4 \cdot \sqrt{2})$ bounding box; see Figure~\ref{fig:case_tree}. Here we have two cases.

\case{3.2.1.2.1}{There does not exist an empty $3$-brick} 
\begin{align*}
& A(empty) \leq \sum_{i \geq 1 \land i \neq 3} A(B_i) \leq \frac{7}{16} \sqrt{2} \\
& A(sparse) \leq \sum_{i \geq 1} A(B_i) \leq \frac{\sqrt{2}}{2}
\end{align*}
hence $A_{pcs} \geq \sqrt{2} / 64$. Since $L^2 > 1 / 8 > \sqrt{2} / 64$ we finally have
\begin{align*}
& \OPT \geq L + \frac{A_{pcs}}{L} \geq \frac{\sqrt{2}}{4} + \frac{1}{16} \\
& \frac{\ALG}{\OPT} \leq \frac{\frac{1}{2} + \frac{3}{4}\sqrt{2}}{\frac{\sqrt{2}}{4} + \frac{1}{16}} \approx 3.75 < 4.
\end{align*}

\case{3.2.1.2.2}{There exists an empty $3$-brick} 

Thanks to Remark~\ref{remark:free_monotonicity} we can assume $B_0 \spl 1 \spl 1 \spl 2$ to be empty. Here we have two cases.

\case{3.2.1.2.2.1}{There does not exist an empty $4$-brick} 

Here we have two cases.

\case{3.2.1.2.2.1.1}{For $j = 5$ or $j = 6$ there does not exist an empty $j$-brick} 
\begin{align*}
& A(empty) \leq \sum_{i \geq 1 \land i \neq 4, j} A(B_i) \leq \frac{59}{128} \sqrt{2} \quad \text{(worst case is when $j = 6$)}\\
& A(sparse) \leq \sum_{i \geq 1} A(B_i) \leq \frac{\sqrt{2}}{2}
\end{align*}
hence $A_{pcs} \geq 5 / 512 \cdot \sqrt{2}$. Since $L^2 > 1 / 8 > 5 / 512 \cdot \sqrt{2}$ we finally have
\begin{align*}
& \OPT \geq L + \frac{A_{pcs}}{L} \geq \frac{\sqrt{2}}{4} + \frac{5}{128} \\
& \frac{\ALG}{\OPT} \leq \frac{\frac{1}{2} + \frac{3}{4}\sqrt{2}}{\frac{\sqrt{2}}{4} + \frac{5}{128}} \approx 3.98 < 4.
\end{align*}

\case{3.2.1.2.2.1.2}{For $j = 5$ and $j = 6$ there exists an empty $j$-brick} 

Thanks to Remark~\ref{remark:free_monotonicity} we can assume $B_0 \spl 1 \spl 1 \spl 1 \spl 2 \spl 2$ and $B_0 \spl 1\spl 1 \spl 1 \spl 2 \spl 1 \spl 2$ to be empty. Then, we can cut the topmost part of $B_{\geq 0}$ and get a $1 / 2 \times (11/16 \cdot \sqrt{2})$ bounding box; see Figure~\ref{fig:case_tree}. We have 
\begin{align*}
& A(empty) \leq \sum_{i \geq 1 \land i \neq 4} A(B_i) \leq \frac{15}{32} \sqrt{2} \\    
& A(sparse) \leq \sum_{i \geq 1} A(B_i) \leq \frac{\sqrt{2}}{2}
\end{align*}
hence $A_{pcs} \geq \sqrt{2} / 128$. Since $L^2 > 1 / 8 > \sqrt{2} / 128$ we finally have 
\begin{align*}
& \OPT \geq L + \frac{A_{pcs}}{L} \geq \frac{\sqrt{2}}{4} + \frac{1}{32} \\
& \frac{\ALG}{\OPT} \leq \frac{\frac{1}{2} + \frac{11}{16}\sqrt{2}}{\frac{\sqrt{2}}{4} + \frac{1}{32}} \approx 3.83 < 4.
\end{align*}

\case{3.2.1.2.2.2}{There exists an empty $4$-brick} 

Thanks to Remark~\ref{remark:free_monotonicity} we can assume $B_0 \spl 1 \spl 1 \spl 1 \spl 2$ to be empty. Then, we can cut the topmost part of $B_{\geq 0}$ and get a $1 / 2 \times (5/8 \cdot \sqrt{2})$ bounding box; see Figure~\ref{fig:case_tree}. Now it remains to bound $\OPT$, and we just assume $\OPT \geq L \geq \sqrt{2} / 4$, finally 
\[
\frac{\ALG}{\OPT} \leq \frac{\frac{1}{2} + \frac{5}{8}\sqrt{2}}{\frac{\sqrt{2}}{4}} \approx 3.92 < 4.
\]
\case{3.2.2}{$M = B_0 \spl 1$} 

For the rest of the proof let $L$ be the length of the longest of pieces' edges then, according to Remark~\ref{remark:edge_length}, $\sqrt{2} / 4 \leq L \leq \sqrt{2} / 2$. 
We can cut the topmost part of $B_{\geq 0}$ and get a $1 \times (\sqrt{2} / 2 + L)$ bounding box; see Figure~\ref{fig:case_tree}. Here we have two cases.

\case{3.2.2.1}{There exists a $2$-brick in $\D$} 

Thanks to Remark~\ref{remark:edge_length}, we have a piece of width at least $1 / 4$, and combining this with the fact that we have a piece of height $L$, it is apparent that $\OPT \geq 1 /4 + L$; see Figure~\ref{fig:case_tree}. Thus,
\[
\frac{\ALG}{\OPT} \leq \frac{1 + \frac{\sqrt{2}}{2} + L}{\frac{1}{4} + L} \leq 2 + \sqrt{2} < 4 \quad \text{(maximizing over $L \in [\sqrt{2}/4, \sqrt{2}/2]$).}
\]

\case{3.2.2.2}{There does not exist a $2$-brick in $\D$}

Here we have two cases.

\case{3.2.2.2.1}{There does not exist an empty $2$-brick} 
\begin{align*}
& A(empty) \leq \sum_{i \geq 1 \land i \neq 2} A(B_i) \leq \frac{3}{8} \sqrt{2} \\
& A(sparse) \leq \sum_{i \geq 1} A(B_i) \leq \frac{3}{8} \sqrt{2}
\end{align*}
hence $A_{pcs} \geq \sqrt{2} / 16$. Since $L^2 > 1 / 8 > \sqrt{2} / 16$ we finally have
\begin{align*}
& \OPT \geq L + \frac{A_{pcs}}{L} \geq L + \frac{\sqrt{2}}{16 L}\\
& \frac{\ALG}{\OPT} \leq \frac{1 + \frac{\sqrt{2}}{2} + L}{L + \frac{\sqrt{2}}{16 L}} \leq 2 + \sqrt{2} < 4 \quad \text{(maximizing over $L \in [\sqrt{2}/4, \sqrt{2}/2]$).}
\end{align*}

\case{3.2.2.2.2}{There exists an empty $2$-brick} 
Thanks to Remark~\ref{remark:free_monotonicity} we can assume $B_1 \spl 2$ to be empty.
Here we have two cases.

\case{3.2.2.2.2.1}{There exists an empty $3$-brick} 

Thanks to Remark~\ref{remark:free_monotonicity} we can assume $B_1 \spl 1 \spl 2$ to be empty. Then, we can cut the rightmost part of $B_{\geq 0}$ and get a $3/4 \times (\sqrt{2}/2 + L)$ bounding box; see Figure~\ref{fig:case_tree}. Now it remains to bound $\OPT$. We have 
\begin{align*}
& A(empty) \leq \sum_{i \geq 1 } A(B_i) \leq \frac{\sqrt{2}}{2} \\
& A(sparse) \leq \sum_{i \geq 1 \land i \neq 2} A(B_i) \leq \frac{3}{8} \sqrt{2}
\end{align*}
hence $A_{pcs} \geq \sqrt{2} / 32$. Since $L^2 > 1 / 8 > \sqrt{2} / 32$ we finally have
\begin{align*}
& \OPT \geq L + \frac{A_{pcs}}{L} \geq L + \frac{\sqrt{2}}{32 L}\\
& \frac{\ALG}{\OPT} \leq \frac{\frac{3}{4} + \frac{\sqrt{2}}{2} + L}{L + \frac{\sqrt{2}}{32 L}} \leq 3.79 < 4 \quad \text{(maximizing over $L \in [\sqrt{2}/4, \sqrt{2}/2]$).}
\end{align*}

\case{3.2.2.2.2.2}{There does not exist an empty $3$-brick} 
\begin{align*}
& A(empty) \leq \sum_{i \geq 1 \land i \neq 3} A(B_i) \leq \frac{7}{16} \sqrt{2} \\
& A(sparse) \leq \sum_{i \geq 2 \land i \neq 2} A(B_i) \leq \frac{3}{8} \sqrt{2}
\end{align*}
hence $A_{pcs} \geq 3 / 64 \cdot \sqrt{2} $. Since $L^2 > 1 / 8 > 3 / 64 \cdot \sqrt{2}$ we finally have
\begin{align*}
& \OPT \geq L + \frac{A_{pcs}}{L} \geq L + \frac{3 \sqrt{2}}{64 L}\\
& \frac{\ALG}{\OPT} \leq \frac{1 + \frac{\sqrt{2}}{2} + L}{L + \frac{3 \sqrt{2}}{64 L}} \leq 3.82 < 4 \quad \text{(maximizing over $L \in [\sqrt{2}/4, \sqrt{2}/2]$).}
\end{align*}
\else
This proof relies on a careful case analysis of which bricks are contained in $\D$ and is deferred to the full version.
\fi
\end{proof}

\pparagraph{Algorithm using rotations}\label{sec:perirotations}
The algorithm $\smallboxrot$ is almost identical to $\smallboxtrans$, but with the difference that we rotate each piece so that its height is at least its width.

\begin{theorem}\label{thm:perirotalg}
The algorithm $\smallboxrot$ has a competitive ratio of strictly less than \compratrot\ for \perirot.
\end{theorem}

\begin{proof}
\iffull
The analysis of $\smallboxtrans$ carried out in the proof of Theorem~\ref{thm:peritransalg} still holds, in fact all the estimates on $\OPT$ derived from consideration about area are still valid, and the only delicate spot is case (3.2.2.1). In that case we assume to have a piece $p$ having an edge of length $L \in [\sqrt{2} / 4, \sqrt{2}/2]$, and that there exists a $2$-brick in $\D$. Thanks to Remark~\ref{remark:edge_length} there exists a piece $q$ of size $w_q \times h_q$ with $w_q \geq 1 / 4$, moreover we rotate every piece so that $1 / 4 \leq w_q\leq h_q$. Finally, a box that contains both $p$ and $q$ must have size at least $L \times w_q$ or $L \times h_q$, hence
\[
\OPT \geq \min\left\{L + w_q, L + h_q\right\} \geq L + \frac{1}{4}. 
\]
This gives exactly the same bound showed in case (3.2.2.1) and completes the proof.
\else
This proof is deferred to the full version.
\fi
\end{proof}

\subsection{A similar but inferior algorithm}\label{sec:alg:inferior}

Here we consider the algorithm we get by making a slight change to \smallboxtrans.
Suppose that the very first piece $p$ arrives and that a $k$-brick is suitable for $p$.
Instead of placing $p$ in $B_k$ (as \smallboxtrans\ would do), we consider the brick $B_{>k}$ to be a fundamental brick (although in the original algorithm, it was an infinite union of fundamental bricks) and we place $p$ in $B_{>k}$.
Thus, we are never going to use the fundamental bricks $B_{i}$ individually, for $i>k$.
From here on, the algorithm does as \smallboxtrans:
Whenever a new piece arrives, we place it in the first derived brick of the suitable size that has room.
This behavior is similar to the algorithm for the problem \sqArea\ that was described by Fekete and Hoffmann~\cite{DBLP:journals/algorithmica/FeketeH17}.
That problem is studied in more detail in Section~\ref{sec:boundedaspect}, and for that problem, the algorithm seems to be no worse than ours.

Interestingly, the following theorem together with Theorem~\ref{thm:peritransalg} implies that the modified algorithm is worse for the problem \peritrans.

\begin{figure}
\centering
\includegraphics[page=6]{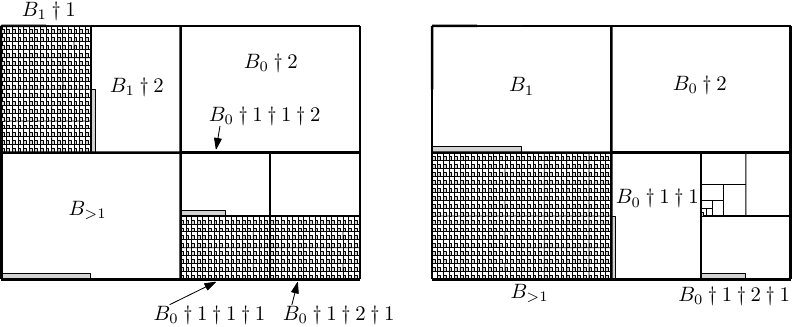}
\caption{Left: A configuration produced by the modified version of \smallboxtrans.
Right: The configuration produced by the original algorithm \smallboxtrans.}
\label{fig:modalg}
\end{figure}

\begin{theorem}
The modified version of \smallboxtrans\ has a competitive ratio of at least $4$ for the problem \peritrans.
\end{theorem}

\begin{proof}
For any $\eps'>0$, we can make an instance realizing a competitive ratio of more than $4-\eps'$ as follows.
Figure~\ref{fig:modalg} shows the packing produced by the modified and the original algorithm.
We first give the algorithm the rectangle $(1/2\sqrt 2+\eps)\times \eps$ for an infinitesimal $\eps>0$.
The rectangle is placed in $B_{>1}$ by the modified algorithm.
For a large odd integer $k$, we then feed the algorithm with small rectangles of size $(\sqrt 2^{-k-1}+\eps)\times (\sqrt 2^{-k}+\eps)$ until $B_1\spl 1$ has been completely split into $(k-2)$-bricks, each of which contains one small rectangle.
We now give the algorithm a piece of size $\eps\times (1/4+\eps)$, which is placed in $B_1\spl 2$.
We again give the algorithm many small rectangles until $B_0\spl 1\spl 1\spl 1$ has been split into $(k-2)$-bricks.
Now follows a rectangle of size $(1/4\sqrt 2+\eps)\times\eps$, which is placed in $B_0\spl 1\spl 1\spl 2$.
Finally, we fill $B_0\spl 1\spl 2\spl 1$ with small rectangles.

Note that as $k\longrightarrow \infty$, the bounding box of the produced packing converges to $B_{\geq 0}$, so it has a perimeter as $B_{-1}$.
On the other hand, observe that as $\eps\longrightarrow 0$, we have $\Sigma\longrightarrow A(B_1)/4=A(B_3)$, since the small rectangles fill out bricks with a total area of $A(B_1)$ and with density $1/4$.
In the limit, all the pieces can actually be packed into $B_3$, so $\OPT$ is at most the perimeter of $B_3$.
But the perimeter of $B_{-1}$ is $4$ times that of $B_3$, which finishes the proof.
\end{proof}

\subsection{Lower bounds}\label{sec:lowerB}

\begin{lemma}
Consider any algorithm $A$ for the problem $\peritrans$.
Then the competitive ratio of $A$ is at least $4/3$.
\end{lemma}

\begin{proof}
We first feed $A$ with two unit squares.
Let the bounding box of the two squares have size $a\times b$ and suppose without loss of generality that $a\leq b$.
Then $a\geq 1$ and $b\geq 2$.
We now give $A$ a rectangle of size $2\times\eps$ for a small value $\eps>0$.
The produced packing has a bounding box of perimeter more than $8$, whereas the optimal has perimeter $6+2\eps$.
Therefore, the competitive ratio is $\frac 8{6+2\eps}=\frac 4{3+\eps}$.
By letting $\eps\longrightarrow 0$, we get that the ratio is at least $4/3$.
\end{proof}

\begin{lemma}
Consider any algorithm $A$ for the problem $\perirot$.
Then the competitive ratio of $A$ is at least $5/4$.
\end{lemma}

\begin{proof}
We first feed $A$ with three unit squares.
Let the bounding box of the three squares have size $a\times b$ and suppose without loss of generality that $a\leq b$.
Suppose first that $b<3$.
Then we must have $a\geq 2$ for the box to contain the squares.
We then give the algorithm the rectangle $\eps\times 3$ for a small value $\eps>0$.
The produced packing has a bounding box of size at least $(2+\eps)\times 3$ and perimeter more than $10$, while the optimal solution has size $(1+\eps)\times 3$ and perimeter $8+2\eps$.

On the other hand, if $b\geq 3$, we give the algorithm one more unit square.
The produced packing has a bounding box of size at least $2\times 3$ or at least $1\times 4$, and thus perimeter at least $10$, while the optimal packing has size $2\times 2$ and perimeter $8$.

We get that the competitive ratio is at least $\frac{10}{8+2\eps}=\frac 5{4+\eps}$, and by letting $\eps\longrightarrow 0$, we get that the ratio is at least $5/4$.
\end{proof}

\section{Area versions}\label{sec:area}

\subsection{General lower bounds}\label{sec:lower:general}

In this section we show that, if we allow pieces to be arbitrary rectangles, we cannot bound the competitive ratio for neither $\areatrans$ nor $\arearot$ as a function of the area $\OPT$ of the optimal packing. However we will be able to bound the competitive ratio as a function of the total number $n$ of pieces in the stream. 

\begin{lemma}\label{thm:nooptbound}
Consider any algorithm $A$ solving $\areatrans$ or $\arearot$ and let any $m\in\N$ and $p\in\mathbb{R}$ be given.
There exists a stream of $n = m^2 + 1$ rectangles such that
(i) the rectangles can be packed into a bounding box of area $2p^2$, and
(ii) algorithm $A$ produces a packing with a bounding box of area at least~$mp^2$.
% Hence, $A$ has an asymptotic competitive ratio of at least $\Omega(\sqrt{n})$ and there exists no function $f$ such that $A$ has an absolute or asymptotic competitive ratio of $f(\OPT)$.
\end{lemma}

\begin{proof}
We first feed $A$ with $m^2$ rectangles of size $p \times \frac p{m^2}$.
These rectangles have total area $p^2$.
Let $a\times b$ be the size of the bounding box of the produced packing.

Suppose first that $a\geq \frac p m$ and $b\geq \frac p m$ hold.
We then feed $A$ with a long rectangle of size $pm^2 \times \frac p{m^2}$.
The produced packing has a bounding box of area at least $\frac pm \cdot pm^2 =mp^2$.
The optimal packing is to pack the $m^2$ small rectangles along the long rectangle, which would produce a packing with bounding box of size $ pm^2\times \frac{2p}{m^2}=2p^2$.

Otherwise, we must have $b > pm$ or $a > pm$, since $a b \geq p^2$.
We then feed $A$ with a square of size $p\times p$.
The produced packing has a bounding box of area at least $p\cdot pm =mp^2$.
The optimal packing is obtained stacking the $m^2$ thin rectangles on top of the big square, which produces a packing with bounding box of size $p\times 2p = 2p^2$.
\end{proof}

\begin{corollary}\label{cor:nofunc}
Let $A$ be an algorithm for $\areatrans$ or $\arearot$.
Then $A$ does not have an asymptotic, and hence also absolute, competitive ratio which is a function of $\OPT$.
\end{corollary}

\begin{proof}
Let $f$ be any function of $\OPT$.
For any value $\OPT=c$, we choose $p\mydef \sqrt{c/2}$.
We now choose $m>2f(c)$ and obtain that the competitive ratio is at least $\frac{mp^2}{2p^2}=m/2>f(c)=f(\OPT)$.
\end{proof}

\begin{corollary}\label{cor:sqrtn}
Let $A$ be an algorithm for $\areatrans$ or $\arearot$.
If $A$ has an asymptotic competitive ratio of $f(n)$, where $n=|L|$ is the number of pieces in the stream, then $f(n)=\Omega(\sqrt n)$.
This holds even when all edges of the pieces are required to have length at least $1$.
\end{corollary}

\begin{proof}
We choose $p\mydef m^2$.
Then all edges have length at least $1$, and the competitive ratio is at least $\frac{mp^2}{2p^2}=m/2=\Omega(\sqrt n)$.
Here, $\OPT$ can be arbitrarily big by choosing $m$ big enough, so it is a lower bound on the asymptotic competitive ratio.
\end{proof}

\subsection{Algorithms for arbitrary pieces}\label{sec:alg:arbitrary}
In this section we provide algorithms that solve $\areatrans$ and $\arearot$ with a competitive ratio of $O(\sqrt{n})$, where $n$ is the total number of pieces. Thus we match the bounds provided in the previous section.

We first describe the algorithm $\dynamicboxtrans$ that solves $\areatrans$.  
We assume to receive a stream of pieces $p_1, \dots, p_n$ of unknown length $n$, such that piece $p_i$ has size $w_i\times h_i$.
For each $k\in\Z$, we define a rectangular box $B_k$ with a size varying dynamically. After pieces $p_1, \dots, p_j$ have been processed $B_k$ has size $2^k\times T_j$, where $T_j \mydef H_j \sqrt{j} + 7H_j$ and $H_j \mydef \max_{i = 1, \dots, j} h_i$.
We place the boxes with their bottom edges on the $x$-axis and in order such that the right edge of $B_{k-1}$ is contained in the left edge of $B_k$; see Figure~\ref{fig:dynboxes}.
Furthermore, we place the lower left corner of box $B_0$ at the point $(1,0)$. It then holds that all the boxes are to the right of the point $(0,0)$.

\begin{figure}
\centering
\includegraphics[page=6]{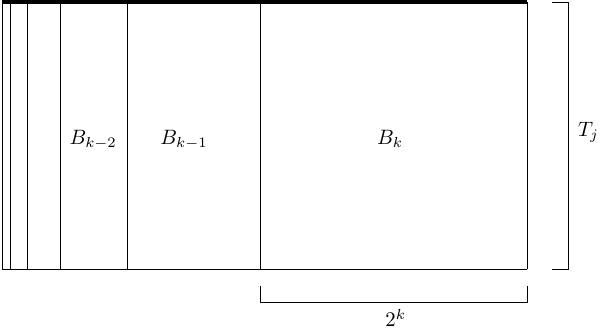}
\caption{The algorithm $\dynamicboxtrans$ packs pieces into the boxes $B_k$ that form a row. Every box has height $T_j$ that is dynamically updated.}
\label{fig:dynboxes}
\end{figure}

We say that the box $B_k$ is \emph{wide enough} for a piece $p_i = w_i \times h_i$ if $w_i \leq 2^k$. If a box $B_k$ is wide enough for $p_i$, we can pack $p_i$ in $B_k$ using the online strip packing algorithm $\NFS_k$ that packs rectangles into a strip of width $2^k$.
The algorithm $\NFS_k$ is the \emph{next-fit shelf algorithm} first described by Baker and Schwartz~\cite{baker1983shelf}.
The algorithm packs pieces in \emph{shelves} (rows), and each shelf is given a fixed height of $2^j$ for some $j\in\Z$ when it is created; see Figure~\ref{fig:nfsalg}.
The width of each shelf is $2^k$, since this is the width of the box $B_k$.

\begin{figure}
\centering
\includegraphics[page=5]{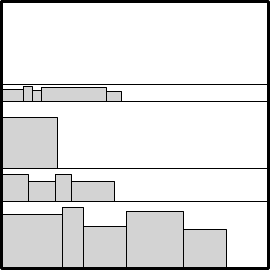}
\caption{A packing produced by the next-fit shelf algorithm using four shelves.}
\label{fig:nfsalg}
\end{figure}

A piece of height $h$, where $2^{j-1} <h \leq 2^j$, is packed in a shelf of height $2^j$.
We divide the shelves into two types.
If the total width of pieces in a shelf is more than $2^{k-1}$ we call that shelf \emph{dense}, otherwise we say it is \emph{sparse}. 
The algorithm $\NFS_k$ places each piece as far left as possible into the currently sparse shelf of the proper height.
If there is no sparse shelf of this height or the sparse shelf has not room for the piece, a new shelf of the appropriate height is created on top of the top shelf, and the piece is placed there at the left end of this new shelf. This ensures that at any point in time there exists at most one sparse shelf for each height $2^j$.

If we allow the height of the box $B_k$ to grow large enough with respect to shelves' heights, the space wasted by sparse shelves becomes negligible and we obtain a constant density strip packing, as stated in the following lemma.
\begin{lemma}\label{lem:constdensity}
Let $\widetilde{H}$ be the total height of shelves in $B_k$, and $H_{max}$ be the maximum height among pieces in $B_k$. If $\widetilde{H} \geq 6 H_{max}$, then the pieces in $B_k$ are packed with density at least $1 / 12$.
\end{lemma}

\begin{proof}
Let $2^{m-1} < H_{max} \leq 2^m$, so that $\widetilde{H} \geq 3 \cdot 2^m$. 
For each $i \leq m$ we have at most one sparse shelf of height $2^i$ and each shelf of $B_k$ has height at most $2^m$, hence the total height of sparse shelves is at most $\sum_{i\leq m} 2^i=2^{m+1}$, so the total height of dense shelves is at least $\widetilde{H} - 2^{m+1} \geq \widetilde{H} / 3$.
Thus, the total area of the dense shelves is at least $2^{k} \cdot \widetilde{H} / 3$.

Consider a dense shelf of height $2^i$.
Into that shelf, we have packed pieces of height at least $2^{i-1}$, and the total width of these pieces is at least $2^{k-1}$.
Hence, the density of pieces in the shelf is at least $1/4$.
Therefore, the total area of pieces in $B_k$ is at least $2^{k} \cdot \widetilde{H} / 12$.
On the other hand, the area of the bounding box is $2^k \cdot \widetilde{H}$, that yields the desired density.
\end{proof}

Now we are ready to describe how the algorithm works. When the first piece $p_1$ arrives, let $2^{k-1} < w_1 \leq 2^k$, then we pack it in the box $B_k$ according to $\NFS_k$ and define $B_k$ to be the \emph{active box}.
Suppose now that $B_i$ is the active box when the piece $p_j$ arrives, first we update the value of the threshold $T_{j-1}$ to $T_j$, then we have two cases. If $w_j > 2^i$ we choose $\ell$ such that $2^{\ell-1} < w_j \leq 2^\ell$, pack $p_j$ in $B_\ell$ and define $B_\ell$ to be the active box. Else, $B_i$ is wide enough for $p_j$ and we try to pack $p_j$ into $B_i$. Since $B_i$ has size $2^i \times T_j$ it may happen that $\NFS_i$ exceeds the threshold $T_j$ while packing $p_j$, generating an overflow. In this case, instead of packing $p_j$ in $B_i$, we pack $p_j$ into $B_{i+1}$ and define that to be the active box.

\begin{theorem}\label{thm:sqrtntrans}
The algorithm $\dynamicboxtrans$ has an absolute competitive ratio of $O(\sqrt{n})$ for the problem $\areatrans$ on a stream of $n$ pieces.
\end{theorem}
\begin{proof}
First, define $\Sigma_j$ as the total area of the first $j$ pieces, $W \mydef \max_{i = 1, \dots, n} w_i$ and recall that $H_j = \max_{i = 1, \dots, j} h_i$ and $T_j = H_j\sqrt{n} + 7H_j$. 
Let $B_k$ be the last active box, so that we can enclose all the pieces in a bounding box of size $2^{k+1} \times T_n$, and bound the area returned by the algorithm as $\ALG = O(2^k H_n \sqrt{n})$. On the other hand we are able to bound the optimal offline packing as $\OPT = \Omega(\Sigma_n + WH_n)$.

If the active box never changed, then we have $2^k < 2W$ that implies $\ALG = O(WH_n\sqrt{n}) = \OPT  \cdot O(\sqrt{n})$. 
Otherwise, let $B_\ell$ be the last active box before $B_k$, and $p_j$ be the first piece put in $B_k$. Here we have two cases.

\case{1}{$w_j > 2^\ell$} In this case we have $2^k < 2W$ that implies $\ALG = O(WH_n\sqrt{n}) = \OPT  \cdot O(\sqrt{n})$.

\case{2}{$w_j \leq 2^\ell$} In this case we have $k = \ell + 1$. Denote with $\widetilde{H_i}$ the total height of shelves in $B_i$. Then we have $\widetilde{H_\ell} \geq T_j - H_j = H_j\sqrt{n} + 6H_j$, otherwise we could pack $p_j$ in $B_\ell$. Thus, we can apply Lemma~\ref{lem:constdensity} and conclude that the box $B_\ell$ of size $2^\ell \times T_j$ is filled with constant density. Here we have two cases.

\case{2.1}{$\widetilde{H_k} \leq T_j$} In this case we have $\ALG = O(2^k T_j)$ and, thanks to the constant density packing of $B_\ell$ we have $\Sigma_j = \Theta (2^\ell \widetilde{H_\ell}) = \Theta(2^k T_j)$. Since $\OPT \geq \Sigma_j$, we get $\ALG = O(\OPT)$.

\case{2.2}{$\widetilde{H_k} > T_j$} In this case we have $\ALG = O(2^k \widetilde{H_k})$. Moreover, $\widetilde{H_k} = O(H_n + \Sigma_n / 2^k)$, in fact if $2^{s-1} < H_n \leq 2^s$, then the total height of sparse shelves is $\sum_{i \leq s} 2^i = 2^{s+1} = O(H_n)$. Furthermore, dense shelves are filled with constant density, therefore their total height is at most $O(\Sigma_n / 2^k)$. Finally, we need to show that $2^k = O(W\sqrt{n})$. Thanks to the constant density packing of $B_\ell$, we have $2^k H_j \sqrt{j} = O(2^\ell T_j) = O(\Sigma_j)$. We can upper bound the size of every piece $p_i$ for $i\leq j$ with $W \times H_j$ and obtain $\Sigma_j \leq n \cdot WH_j$. Plugging it in the previous estimate and dividing both sides by $H_j\sqrt{n}$ we get $2^k = O(W\sqrt{n})$.
Now we have $\ALG = O(2^k \widetilde{H_k}) = O(2^k H_n + \Sigma_n) = O(WH_n\sqrt{n} + \Sigma_n) = \OPT \cdot O(\sqrt{n})$.
\end{proof}

The algorithm $\dynamicboxrot$ is obtained from $\dynamicboxtrans$ with a slight modification: before processing any piece $p_i$ we rotate it so that $w_i \leq h_i$. In this way, it still holds that $\OPT = \Omega(\Sigma_n + WH_n )$ and the proof of Theorem~\ref{thm:sqrtntrans} works also for the following.

\begin{theorem}\label{thm:sqrtnrot}
The algorithm $\dynamicboxrot$ has an absolute competitive ratio of $O(\sqrt{n})$ for the problem $\arearot$ on a stream of $n$ pieces.
\end{theorem}

\subsection{Bounded aspect ratio}\label{sec:boundedaspect}

In this section, we will consider the special case where the aspect ratio of all pieces is $\alpha=1$, i.e., all the pieces are squares.
Furthermore, we will measure the size of the packing as the area of the minimum axis-parallel bounding \emph{square}, and we call the resulting problem \sqArea.
Since we get a constant competitive ratio in this case, it follows that for other values of $\alpha$ and when allowing the bounding box to be a general rectangle, one can likewise achieve a constant competitive ratio.
We first give a lower bound.

\begin{lemma}
Consider any algorithm $A$ for the problem $\sqArea$.
Then the competitive ratio of $A$ is at least $16/9$.
\end{lemma}

\begin{proof}
We first give $A$ four $1\times 1$ squares.
Let the bounding square have size $\ell \times \ell$.
If $\ell \geq 3$, the bounding square of the four $1\times 1$ squares has size at least $3\times 3$, while the optimal packing has size $2\times 2$, which gives ratio at least $9/4$.
Otherwise, if $\ell<3$, we give a $2\times 2$ square and we will prove that the bounding square has size at least $4\times 4$ while the optimal packing fits in a $3\times 3$ square, so the ratio is at least $16/9$. 

Let us assume by contradiction that there exists a $(4 - \varepsilon) \times (4 - \varepsilon)$ bounding square containing both a $2\times 2$ square and four $1\times 1$ squares, with the additional hypothesis that the $1\times 1$ squares fit in a $(3-\delta) \times (3-\delta)$ bounding box. We refer to notation in Figure~\ref{fig:2x2in4x4} (left) and notice that we have $a < 1$ or $b < 1$, and analogously $c < 1$ or $d < 1$. Without loss of generality, we can assume $a, d < 1$. Hence, starting from the configuration in Figure~\ref{fig:2x2in4x4} (left) we can drag the $2\times 2$ square to the bottom left corner and obtain the configuration in Figure~\ref{fig:2x2in4x4} (right), that still fulfill the hypotheses we assumed by contradiction.  

\begin{figure}
\centering
\includegraphics[page=9]{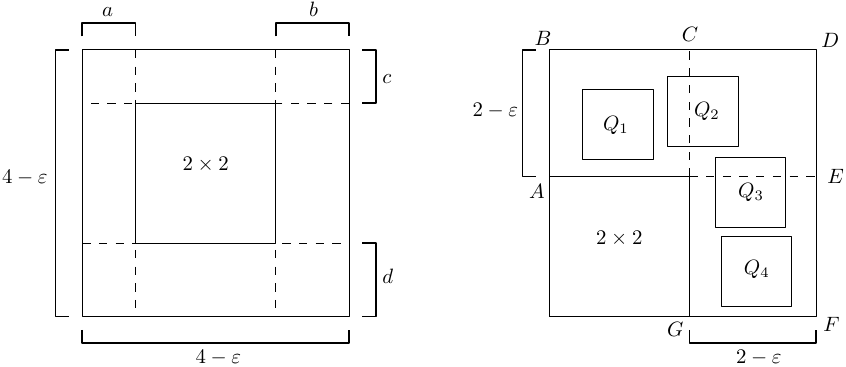}
\caption{Left: A $2\times 2$ square inside a bounding square having edges shorter than $4$. Right: The $2\times 2$ square has been dragged in the bottom left corner of the bounding square. Four $1\times1$ squares $Q_1, \dots, Q_4$ are placed within the bounding square.}
\label{fig:2x2in4x4}
\end{figure}

% \begin{figure}
% \centering
% \includegraphics[page=8, scale=0.5]{figs2.pdf}
% \caption{The $2\times 2$ square has been dragged in the bottom left corner of the bounding square. Four $1\times1$ squares $Q_1 \dots Q_4$ are placed within the bounding square.}
% \label{fig:dragged2x2}
% \end{figure}

From now on we employ the notation of Figure~\ref{fig:2x2in4x4} (right). Let $(x_i, y_i)$ be the coordinates of the bottom left corner of square $Q_i$. Stating that $Q_i$ and $Q_j$ are disjoint is equivalent to $\max\{|x_i - x_j|, |y_i - y_j|\} \geq 1$. Consider now the two rectangular regions $ABDE$ and $GCDF$: note that each of them can contain at most two squares.  
Indeed, given $Q_i$ and $Q_j$ completely contained in $ABDE$, it holds $|y_i - y_j| \leq 1 - \varepsilon$ thus $|x_i - x_j| \geq 1$. If three squares $Q_1, Q_2, Q_3$ are completely contain in $ABDE$ then we have, without loss of generality, $x_1 \leq x_2 - 1 \leq x_3 -2$ and the minimal bounding square of $Q_1, Q_2, Q_3$ has size at least $3\times 3$, that gives a contradiction. The same holds for $GCDF$.

Finally, every $Q_i$ is either fully contained in $ABDE$ or $GCDF$ hence, without loss of generality, we can assume that $Q_1, Q_2$ are contained in $ABDE$ and $Q_3, Q_4$ are contained in $GCDF$. This implies that $x_1 \leq x_2 -1$ and $y_4 \leq y_3 - 1$, again without loss of generality.
Observe that $x_2 \leq x_3 + 1 - \varepsilon$ and $y_3 \leq y_2 + 1 - \varepsilon$. 
$Q_2$ and $Q_3$ are disjoint, using the previous characterization we have two cases. First, $|x_2 - x_3| \geq 1$ and thanks to the observation above it cannot be $x_2 > x_3$, therefore we have $x_1 \leq x_2 - 1 \leq x_3 -2$. Else, $|y_2 - y_3| \geq 1$ and thanks to the observation above we have $y_4 \leq y_3 - 1 \leq y_2 -2$. In both cases that gives a contradiction since we cannot pack all $Q_i$s in a $(3-\delta)\times (3-\delta)$ bounding square.
\end{proof}

We are now going to analyze the competitive ratio of the algorithm \smallboxtrans\ (in fact, the algorithm \smallboxrot\ has the exact same behavior when the pieces are squares).
Note that a brick can never contain more than one piece.
The algorithm is almost the same as the one described by Fekete and Hoffmann~\cite{DBLP:journals/algorithmica/FeketeH17}.
The slight difference is addressed in Section~\ref{sec:alg:inferior} and it is shown there that the behavior as described by Fekete and Hoffmann makes a worse algorithm for the problem \peritrans.
However, even though the two algorithms will not always produce identical packings for the problem \sqArea, the analysis of the following theorem seems to hold for both versions, so for the problem \sqArea, the algorithms are equally good.

\begin{theorem} \label{thm:square_in_square}
The algorithm $\smallboxtrans$ has a competitive ratio of $\compratsq$ for \sqArea.
The analysis is tight.
\end{theorem}

\begin{proof}
Suppose a stream of squares have been packed by $\smallboxtrans$, and let $\alg$ be the area of the bounding square of the resulting packing.
Let $B_k$ be the largest elementary brick in which a square has been placed.
Suppose without loss of generality that $k=0$, so that $B_k$ has size $1\times 1/\sqrt 2$ and $B_{\geq k}$, which contains all the packed squares, has size $1\times\sqrt 2$.

\begin{figure}
\centering
\includegraphics[page=2]{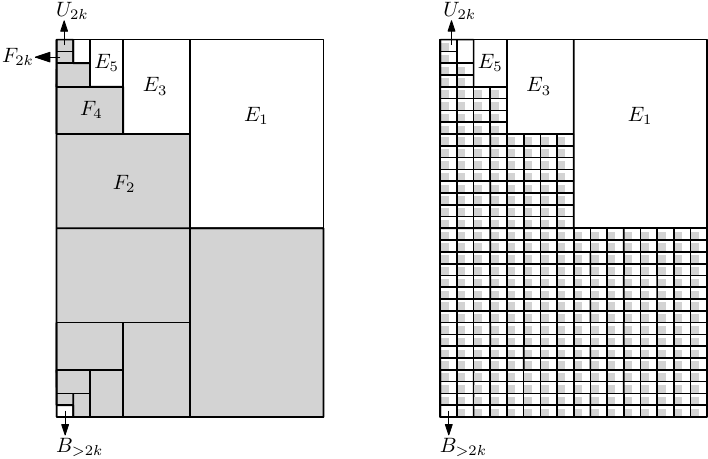}
\caption{Left: A $2k$-packing.
The grey bricks are non-empty and may have been split into smaller bricks.
Right: The $2k$-packing produced by \smallboxtrans\ when providing the algorithm with enough copies of the square $S_k$ (the small grey squares), showing that the competitive ratio can be arbitrarily close to $6$.}
\label{fig:seqLower}
\end{figure}

We now recursively define a type of packing that we call a $2k$-packing, for a non-negative integer $k$; see Figure~\ref{fig:seqLower} (left).
As $k$ increases, so do the requirements to a $2k$-packing, in the sense that a $(2k+2)$-packing is also a $2k$-packing, but the other way is in general not the case.
Define $F_0\mydef B_{\geq 1}$ and $U_0\mydef B_0$.
A packing is a $0$-packing if pieces have been placed in $U_0$ (the brick $U_0$ may or may not have been split in smaller bricks).
Hence, the considered packing is a $0$-packing by the assumption that a piece has been placed in $B_0$.
Suppose that we have defined a $2k$-packing for some integer $k$.
A $(2k+2)$-packing is a $2k$-packing with the additional requirements that
\begin{itemize}
\item
the brick $U_{2k}$ has been split into $L\mydef U_{2k}\spl 1$ and $E_{2k+1}\mydef U_{2k}\spl 2$,

\item
the right brick $E_{2k+1}$ is empty,

\item
the left brick $L$ has been split into $F_{2k+2}\mydef L\spl 1$ and $U_{2k+2}\mydef L\spl 2$, and

\item
%a square has been placed in $U_{2k+2}$, and thus also in $F_{2k+2}$ (these bricks may have been further split into smaller bricks).
$U_{2k+2}$ is non-empty, and thus also $F_{2k+2}$ is non-empty.
\end{itemize}

The symbols $U_j,E_j,F_j$ have been chosen such that the brick is a $j$-brick, i.e., the index tells the size of the brick.

Consider a $2k$-packing.
It follows from the definition that along the top edge of $B_{\geq 0}$ from the right corner $(1,\sqrt 2)$ to the left corner $(0,\sqrt 2)$, we meet a sequence $E_{1},E_{3},\ldots,E_{2k-1}$ of empty bricks of decreasing size, and finally meet a non-empty brick $U_{2k}$ which may have been split into smaller bricks.

% \emph{Claim 1:}
% The area of the squares in the bricks $F_0,F_{2},\ldots,F_{2k}$ in a $2k$-packing is at least $\frac{1 - 4^{-k}}3$.

%We are now able to state the claim which will lead to the proof of Theorem~\ref{thm:square_in_square}.

\iffull
\else
In the full version the following claim is proven.
\fi

\begin{claim} \label{claim:2k-packing}
If the packing is a $2k$-packing and not a $(2k+2)$-packing, then $\alg /\OPT < 6$.
\end{claim}
 Since we pack a finite number of squares, the produced packing is a $2k$-packing but not a $(2k+2)$-packing for some sufficiently large $k$, so Claim~\ref{claim:2k-packing} implies Theorem~\ref{thm:square_in_square}.
%Hence, the competitive ratio of $\smallboxtrans$ is at most $6$.

\iffull
Let us now prove Claim~\ref{claim:2k-packing}.
We first compute the area of the brick $U_{2k}$ and the total areas of the bricks $F_0,F_2,\ldots,F_{2k}$, as these areas will be used often:
\begin{align}
u_k & \mydef |U_{2k}| =2^{-2k}/\sqrt 2 \\
f_k & \mydef \sum_{i=0}^k |F_{2i}| =\frac {2|B_{\geq 0}|-u_k}{3}=\frac{4-4^{-k}}{3\sqrt 2} \label{eq:fsum}.
\end{align}

%We prove the claim by the following case analysis.
\begin{itemize}
\item[1)]
Suppose first that $U_{2k}$ has not been split into smaller bricks.
Then, since $U_{2k}$ is non-empty by assumption, we know that $U_{2k}$ contains a square $S$ of size $s\times s$ where $s\in(s_l,s_h]=\left(\sqrt 2^{-2k-2} , \sqrt 2^{-2k-1}\right]$.
Since the bricks $E_{1},E_{3},\ldots,E_{2k-1}$ are all empty, we get that the upper edge of the bounding square coincides with the upper edge of $S$, and we thus have
\[
\alg\leq \alg(s)\mydef (\sqrt 2-(\sqrt 2^{-2k-1}-s))^2.
\]
The largest empty brick in the bricks $F_{2i}$ can have size $|U_{2k}|/2$, so the total size of empty bricks in $F_0,F_2,\ldots,F_{2k}$ is $|U_{2k}|$. Moreover, the density of squares into bricks is at least $1 / 2 \sqrt{2}$ and by~\eqref{eq:fsum}, we get that
\[
\OPT\geq \OPT(s)\mydef \frac{f_k-u_k}{2\sqrt 2}+s^2= \frac{1 - 4^{-k}}3+s^2.
\]

In the case that $k=0$, we get
\[
\frac{\alg}{\OPT}\leq \frac{\alg(s)}{\OPT(s)}= \frac{2s\sqrt 2 + 2s^2 + 1}{2s^2}.
\]
A simple analysis shows that the fraction is largest when $s=s_l$, so we get the bound
\[
\frac{\alg}{\OPT}\leq \frac{2s_l\sqrt 2 + 2s_l^2 + 1}{2s_l^2}=3+2\sqrt 2<5.83
\]

Suppose now that $k>0$.
We divide into two cases of whether $s$ is in the lower or the upper half of the range $(s_l,s_h]$.
For the lower half, that is, $s\in (s_l,\frac{s_l+s_h}2]$, we get
\[
\frac{\alg}{\OPT}\leq \frac{\alg(\frac{s_l+s_h}2)}{\OPT(s_l)}= \frac{96\cdot 4^k + (24\sqrt 2-48)\cdot 2^{k} - 6\sqrt 2 + 9}{16\cdot 4^k - 4}.
\]
It is straightforward to check that $(24\sqrt 2-48)\cdot 2^{k} - 6\sqrt 2 + 9<6\cdot(-4)$ for all $k\geq 1$, so it follows that the ratio is less than $6$.

For the upper half, that is, $s\in [\frac{s_l+s_h}2,s_h]$, we get
\[
\frac{\alg}{\OPT}\leq  \frac{\alg(s_h)}{\OPT(\frac{s_l+s_h}2)}= \frac{96\cdot 4^k}{16\cdot 4^k + 6\sqrt 2 - 7}.
\]
As $6\sqrt 2 - 7>0$, the ratio is less than $6$.
% We conclude that if the ratio is at least $6$, then $U_{2k}$ has been split into smaller bricks.

\item[2)]
We now assume that $U_{2k}$ has been split into a $L$ and $E_{2k+1}$, which are the left and right halfs of $U_{2k}$, respectively.

\begin{itemize}
\item[2.1)]
We first suppose that $E_{2k+1}$ is not empty.
This implies that there is no empty $(2k+1)$-brick in $F_0,F_2,\ldots,F_{2k},U_{2k}$.
Hence, each empty brick in the bricks $F_0,F_2,\ldots,F_{2k},U_{2k}$ is a $(2k+2)$-brick or smaller, so these empty bricks have total size at most $u_k/2$.
We then get
\[
\OPT\geq \frac{f_k+u_k-u_k/2}{2\sqrt 2}=\frac 13 + \frac{4^{-k}}{24}>\frac 13.
\]

Since $\alg\leq 2$, it follows that $\frac{\alg}{\OPT}<6$.

\item[2.2)]
We now suppose that $E_{2k+1}$ is empty.
% We can therefore conclude that $E_{2k+1}$ is empty.

\begin{itemize}
\item[2.2.1)]
Suppose now that $L$ has not been split into smaller bricks.
Then $L$ contains a square $S$ of size $s\times s$ for $s\in (s_l,s_h]=\left(\sqrt 2^{-2k-3},\sqrt 2^{-2k-2}\right]$.
As in case 1, we get
\[
\alg\leq \alg(s)\mydef (\sqrt 2-(\sqrt 2^{-2k-1}-s))^2.
\]

Note that there is no empty $(2k+1)$-brick in the bricks $F_0,F_2,\ldots,F_{2k}$, so these bricks contain a total area of at most $u_k/2$ empty bricks.
We then get
\[
\OPT\geq \OPT(s)\mydef \frac{f_k-u_k/2}{2\sqrt 2}+s^2.
\]

We then get the bound
\[
\frac{\alg}{\OPT}\leq  \frac{\alg(s_h)}{\OPT(s_l)}= {\frac {24\cdot {4}^{k}+(12\sqrt 2-24)\cdot {2}^{k}-6\,\sqrt {2}+9}{4\cdot {4}^{k}-1}  }.
\]
Here, it is straightforward to verify that $(12\sqrt 2-24)\cdot {2}^{k}-6\,\sqrt {2}+9<6\cdot (-1)$ for all $k\geq 0$, and hence the ratio is less than $6$.

\item[2.2.2)]
We now assume that $L$ has been split into $F_{2k+2}$ and $U_{2k+2}$, which are the bottom and top parts, respectively.

\begin{itemize}
\item[2.2.2.1)]
Suppose that $U_{2k+2}$ is empty.
Since also $E_1,E_3,\ldots,E_{2k+1}$ are empty, we get that $\alg\leq (\sqrt 2-\sqrt 2^{-2k-1}/2)^2$.

Note that each empty bricks in the bricks $F_0,F_2,\ldots,F_{2k+2}$ can have size at most $u_k/8$, so the total size of the empty bricks is at most $u_k/4=u_{k+1}$, and we get
\[
\OPT\geq \frac{f_{k+1}-u_{k+1}}{2\sqrt 2}.
\]

We therefore get
\[
\frac{\alg}{\OPT}\leq {\frac {48\cdot {4}^{k}-24\cdot {2}^{k}+3}{8\cdot {2}^{2\,k}-2}}.
\]
Here, it is straightforward to check that $-24\cdot {2}^{k}+3<6\cdot (-2)$ for all $k\geq 0$, so the ratio is less than $6$.

\item[2.2.2.2)]
We are finally left with the case that $U_{2k+2}$ is not empty.
But then all the requirements are satisfied for the packing to be a $(2k+2)$-packing.

\end{itemize}
\end{itemize}
\end{itemize}
\end{itemize}

We now observe that the analysis is tight.
To this end, we show that for any given $k$ and a small $\eps>0$, we can force the algorithm to produce a $2k$-packing, such that as $k\longrightarrow \infty$ and $\eps\longrightarrow 0$, the ratio $\frac{\alg}{\Sigma}$ tends to $6$, where $\Sigma$ is the total area of the packed squares.
Let $\eps_k\mydef \eps \sqrt 2^{-k}$, $\ell_k\mydef \sqrt 2^{-k}/2+\eps_k$, and let $S_k$ be a square of size $\ell_k\times\ell_k$.
We now feed the algorithm with copies of $S_k$.
This will eventually result in a $2k$-packing, where each non-empty brick is a $2k$-brick; see Figure~\ref{fig:seqLower} (right).
Let $n_k$ be the number needed to produce the $2k$-packing.
The density in each non-empty brick is $\rho_\eps\mydef \frac{|S_k|}{|B_{2k}|}$.
As $\eps\longrightarrow 0$, we get that $\rho_\eps\longrightarrow \frac 1{2\sqrt 2}$.
As $k\longrightarrow\infty$, the area of non-empty bricks converges to $\frac {2|B_{\leq 0}|}{3}=\frac{2\sqrt 2}{3}$.
Hence, we have $\Sigma\longrightarrow \frac 1{2\sqrt 2}\cdot \frac{2\sqrt 2}{3}=\frac 13$.
We then get $\frac{\alg}{\Sigma}\longrightarrow \frac{2}{1/3}=6$.
Furthermore, the optimal packing of the squares is to place them so that their bounding box is a square of size $\lceil \sqrt n_k\rceil\ell_k\times \lceil \sqrt n_k\rceil\ell_k$.
As $k\longrightarrow \infty$, we then have $\frac{\Sigma}{\OPT}\longrightarrow 1$.
Hence, we have $\frac{\alg}{\OPT}\longrightarrow 6$.
\else
Moreover we prove in the full version that this analysis is tight.
The idea is to show that for any given $k$ and a small $\eps>0$, we can force the algorithm to produce a $2k$-packing, such that as $k\longrightarrow \infty$ and $\eps\longrightarrow 0$, the ratio $\frac{\alg}{\OPT}$ tends to $6$.
We use a stream where all pieces are a square $S_k$ of size slightly more than $\sqrt 2^{-k}/2\times \sqrt 2^{-k}/2$; see Figure~\ref{fig:seqLower} (right).
\fi
\end{proof}

\subsection{More lower bounds when edges are long}\label{sec:lower:restricted}

We already saw in Corollary~\ref{cor:sqrtn} that as a function of $n$, the competitive ratio of an algorithm for \areatrans\ or \arearot\ must be at least $\Omega(\sqrt n)$, even when all edges have length $1$.
In this section, we give lower bounds in terms of $\OPT$ for the same case.
Note that the assumption that the edges are long is needed for these bounds to be matched by actual algorithms, since Corollary~\ref{cor:nofunc} states that without the assumption, the competitive ratio cannot be bounded as a function of $\OPT$.

\begin{theorem}\label{thm:areatransbound}
Consider any algorithm $A$ for the problem $\areatrans$ with the restriction that all edges of the given rectangles have length at least $1$.
If $A$ has an asymptotic competitive ratio $f(\OPT)$ as a function of $\OPT$, then $f(\OPT)=\Omega(\sqrt{\OPT})$.
% Any algorithm for the problem $\areatrans$ has an asymptotic competitive ratio of at least $\Omega(\sqrt{\OPT})$, even when all edges have length at least $1$.
\end{theorem}

\begin{remark}
Note that when the edges are long, $\Omega(\sqrt\OPT)=\Omega(\sqrt n)$, so this bound is stronger than the $\Omega(\sqrt n)$ bound of Corollary~\ref{cor:sqrtn}.
\end{remark}

\begin{proof}[Proof of Theorem~\ref{thm:areatransbound}.]
For any $n\in\N$, we do as follows.
We first provide $A$ with $n^2$ unit squares.
Let the bounding box of the produced packing of these squares have size $a\times b$.
Assume without loss of generality that $a\leq b$, so that $b\geq n$.
We now give $A$ the rectangle $n^2\times 1$.
The optimal offline solution to this set of rectangles has a bounding box of size $n^2\times 2$.
The packing produced by $A$ has a bounding box of size at least $n^2\times n=\Omega(\sqrt{\OPT})\cdot \OPT$.
\end{proof}

\begin{theorem}\label{thm:arearotbound}
Consider any algorithm $A$ for the problem $\arearot$ with the restriction that all edges of the given rectangles have length at least $1$.
If $A$ has a competitive ratio $f(\OPT)$ as a function of $\OPT$, then $f(\OPT)=\Omega(\sqrt[4]{\OPT})$.
% Any algorithm for the problem $\arearot$ has an asymptotic competitive ratio of at least $\Omega(\sqrt[4]{\OPT})$, even when all edges have length at least $1$.
\end{theorem}

\begin{proof}
For any $n\in\N$, we do as follows.
We first provide $A$ with $n^2$ unit squares.
Let the bounding box of the produced packing of these squares have size $a\times b$.
Assume without loss of generality that $a\leq b$.
% Then $a\geq n^{1/2}$ or $b\geq n^{3/2}$.
If $a\geq n^{1/2}$, we give $A$ the rectangle $1\times n^2$.
Otherwise, we have $b> n^{3/2}$, and then we give $A$ the square $n\times n$.
In either case, there is an optimal offline solution of area $2n^2$, but the bounding box of the packing produced by $A$ has area at least $n^{5/2}=\Omega(\sqrt[4]{\OPT})\cdot \OPT$.
\end{proof}

\subsection{Algorithms when edges are long}\label{alg:area}

In this section, we describe algorithms that match lower bounds of Section~\ref{sec:lower:restricted}.
We analyze these algorithms under the assumption that we feed them with rectangles with edges of length at least~$1$ (of course, any other positive constant will also work), but we require no bound on the aspect ratio.
Under this assumption, we observe that $\dynamicboxtrans$ has absolute competitive ratio
$O(\sqrt{\OPT})$ for \areatrans.
We then describe the algorithm 
\dynamicboxrotopt, which we prove to have absolute competitive ratio 
$O(\sqrt[4]{\OPT})$ for \arearot.
By Theorems~\ref{thm:areatransbound} and~\ref{thm:arearotbound}, both algorithms are optimal to within a constant factor.
% We do not attempt to determine or minimize the constant hidden in the $O$-notation, as we did for the perimeter problems and the area problem for packing squares.

In previous sections we proved lower bounds of $\Omega(\sqrt{n})$ and $\Omega(\sqrt[4]{\OPT})$ for \arearot. They can be summarized stating that $\arearot$ has a competitive ratio of $\Omega(\max\{\sqrt{n}, \sqrt[4]{\OPT}\})$. The last theorem of this section, describes the algorithm $\dynamicboxrotmin$ that simultaneously matches both lower bounds achieving a competitive ratio of $O(\min\{\sqrt{n}, \sqrt[4]{\OPT}\}$. At a first sight it may seem that this algorithm contradicts the lower bound of $\Omega(\max\{\sqrt{n}, \sqrt[4]{\OPT}\})$; however this simply proves that the \emph{edge cases} that have a competitive ratio of at least $\Omega(\sqrt[4]{\OPT})$ must satisfy $\OPT = O(n^2)$. Likewise, those for which the competitive ratio is at least $\Omega(\sqrt{n})$ satisfy $n = O(\sqrt{\OPT})$.

\pparagraph{Translations only}
Under the long edge assumption, we have $n \leq \OPT$. Therefore, $\dynamicboxtrans$ achieves a competitive ratio of $O(\sqrt{n}) = O(\sqrt{\OPT})$ for $\areatrans$ and matches the bound stated in Theorem~\ref{thm:areatransbound}.

\pparagraph{Rotations allowed}
Now we tackle the $\arearot$ problem and describe the algorithm \dynamicboxrotopt. We define the threshold function $T_j = \Sigma_j^{3/4} + 7H_j$, where $H_j = \max_{i = 1, \dots, j} h_i $ and $\Sigma_j$ is the total area of pieces $p_1, \dots, p_j$. $\dynamicboxrotopt$ is obtained by running $\dynamicboxrot$, as described in Section~\ref{sec:alg:arbitrary}, employing this new threshold $T_j$.

\begin{theorem}\label{thm:simultaneousmatch}
The algorithm \dynamicboxrotopt\ has an absolute competitive ratio of $O(\sqrt[4]{\OPT})$ for the problem $\arearot$, where $\OPT$ is the area of the optimal offline packing.
\end{theorem}
\begin{proof}
This proof is similar to the one of Theorem~\ref{thm:sqrtntrans}. Define $W \mydef \max_{i = 1, \dots, n} w_i$. Recall that in $\dynamicboxrot$ we preprocess every piece $p$ rotating it so the $w_p \leq h_p$, hence $W \leq \sqrt{\Sigma_n}$. Let $B_k$ be the last active box, so that we can enclose all the pieces in a bounding box of size $2^{k+1} \times T_n$, and bound the area returned by the algorithm as $\ALG = O(2^k H_n + 2^k \Sigma_n^{3/4})$. On the other hand we are able to bound the optimal offline packing as $\OPT = \Omega(\Sigma_n + WH_n)$.

If the active box never changed, then we have $2^k < 2W$ that implies $\ALG = O(WH_n + \Sigma_n^{5/4}) = \OPT \cdot O( \sqrt[4]{\OPT})$. 
Otherwise, let $B_\ell$ be the last active box before $B_k$, and $p_j$ be the first piece put in $B_k$. Here we have two cases.

\case{1}{$w_j > 2^\ell$} In this case we have $2^k < 2W$ that implies $\ALG = O(WH_n + \Sigma_n^{5/4}) = \OPT \cdot O(\sqrt[4]{\OPT})$.

\case{2}{$w_j \leq 2^\ell$} In this case we have $k = \ell + 1$. Denote with $\widetilde{H_i}$ the total height of shelves in $B_i$. Then we have $\widetilde{H_\ell} \geq T_j - H_j = \Sigma_j^{3/4} + 6H_j$, otherwise we could pack $p_j$ in $B_\ell$. Thus, we can apply Lemma~\ref{lem:constdensity} and conclude that the box $B_\ell$ of size $2^\ell \times T_j$ is filled with constant density. Here we have two cases.

\case{2.1}{$\widetilde{H_k} \leq T_j$} In this case we have $\ALG = O(2^k T_j)$ and, thanks to the constant density packing of $B_\ell$ we have $\Sigma_j = \Theta (2^\ell \widetilde{H_\ell}) = \Theta(2^k T_j)$. Since $\OPT \geq \Sigma_j$, we get $\ALG = O(\OPT)$.

\case{2.2}{$\widetilde{H_k} > T_j$} In this case we have $\ALG = O(2^k \widetilde{H_k})$. Moreover, $\widetilde{H_k} = O(H_n + \Sigma_n / 2^k)$, in fact if $2^{s-1} < H_n \leq 2^s$, then the total height of sparse shelves is $\sum_{i \leq s} 2^i = 2^{s+1} = O(H_n)$. Furthermore, dense shelves are filled with constant density, therefore their total height is at most $O(\Sigma_n / 2^k)$. Finally, we need to show that $2^k = O(\sqrt[4]{\Sigma_n})$. Thanks to the constant density packing of $B_\ell$, we have $2^k \Sigma_j^{3/4} = O(2^\ell T_j) = O(\Sigma_j)$. Dividing both sides by $\Sigma_j^{3/4}$ we get $2^k = O(\Sigma_j^{1/4})$. In the end notice that, thanks to the long edge hypotheses $H_n \leq \Sigma_n$ and we have $\ALG = O(2^k \widetilde{H_k}) = O(2^k H_n + \Sigma_n) = O(\Sigma_n^{5/4}) = \OPT  \cdot O(\sqrt[4]{\OPT})$.
\end{proof}

So far we managed to match the competitive ratio lower bounds of $\Omega(\sqrt{n})$ and $\Omega(\sqrt[4]{\OPT})$ employing two different algorithms: $\dynamicboxrot$ and $\dynamicboxrotopt$. A natural question is whether is it possible to match the performance of these algorithms simultaneously, having an algorithm that achieves a competitive ratio of $O(\min\{\sqrt{n}, \sqrt[4]{\OPT}\})$. We give an affirmative answer by describing the algorithm $\dynamicboxrotmin$.

Again, we employ the same scheme of $\dynamicboxrot$ with a different threshold function. This time the definition of $T_j$ is slightly more involved. First define
\begin{equation*}
    \widetilde{T}_j = 
    \begin{cases}
   \Sigma_j^{3/4} + 7H_j,   & \text{ if } \Sigma_j < j^2\\
   H_j \sqrt{n} + 7H_j, & \text{ otherwise.}
           \end{cases}
\end{equation*}
Later we will write $\widetilde{T}_j$ as $\widetilde{T}_j =           \mathbbm{1}_{\{\Sigma_{j} < j^2\}} \cdot \Sigma_j^{3/4} +
    \mathbbm{1}_{\{\Sigma_{j} \geq {j}^2\}} \cdot
    H_j \sqrt{n} + 7H_j$.
    We now define
\begin{equation*}
    T_j = \begin{cases}
    0, & \text{ if } j=0 \\
    \max\left\{T_{j-1}, \widetilde{T_j}\right\}, & \text{ if } j \geq 1.
    \end{cases}
\end{equation*}
This two-step definition is necessary for a correct implementation of the algorithm because we must guarantee that $T_j$ does not decrease.

\begin{theorem}\label{thm:sqrt4opt}
When used on the problem $\arearot$, the algorithm $\dynamicboxrotmin$ has an absolute competitive ratio of $O(\min\{\sqrt{n}, \sqrt[4]{\OPT}\})$, where $\OPT$ is the area of the optimal offline packing and $n$ is the total number of pieces in the stream.
\end{theorem}
\begin{proof}
Again, we define $W \mydef \max_{i = 1, \dots, n} w_i$. Recall that in $\dynamicboxrot$ we preprocess every piece $p$ rotating it so the $w_p \leq h_p$, hence $W \leq \sqrt{\Sigma_n}$. Let $B_k$ be the last active box, so that we can enclose all the pieces in a bounding box of size $2^{k+1} \times T_n$. There exists a $n^\prime \leq n$ such that $T_n = \widetilde{T}_{n^\prime}$. We can bound the area returned by the algorithm as

$$\ALG = O\left(2^k \widetilde{T}_{n^\prime}\right) = O \left(2^k H_{n^\prime} + \mathbbm{1}_{\{\Sigma_{n^\prime} < {n^\prime}^2\}} \cdot 2^k \Sigma_{n^\prime}^{3/4}  +  \mathbbm{1}_{\{\Sigma_{n^\prime} \geq {n^\prime}^2\}} \cdot 2^k H_{n^\prime} \sqrt{{n^\prime}}\right).$$ 
We bound the optimal offline packing as $\OPT = \Omega(\Sigma_n + WH_n)$.
If the active box never changed, then we have $2^k < 2W$ that implies 
\begin{align*}
ALG &= O \left(W H_{n} + \mathbbm{1}_{\{\Sigma_{n^\prime} < {n^\prime}^2\}} \cdot W \Sigma_{n^\prime}^{3/4}  +  \mathbbm{1}_{\{\Sigma_{n^\prime} \geq {n^\prime}^2\}} \cdot W H_{n^\prime} \sqrt{{n^\prime}}\right) \\
&= 
O\left(W H_{n} + \mathbbm{1}_{\{\Sigma_{n^\prime} < {n^\prime}^2\}} \cdot \Sigma_{n} \sqrt[4]{\Sigma_{n^\prime}}  +  \mathbbm{1}_{\{\Sigma_{n^\prime} \geq {n^\prime}^2\}} \cdot W H_{n} \sqrt{{n^\prime}}\right) \\
&\leq  
\OPT \cdot O\left( \min\left\{\sqrt[4]{\Sigma_{n^\prime}}, \sqrt{n^\prime}\right\}\right) 
=
\OPT \cdot O\left( \min\left\{\sqrt[4]{\OPT}, \sqrt{n}\right\}\right).
\end{align*}

Otherwise, let $B_\ell$ be the last active box before $B_k$, and $p_j$ be the first piece put in $B_k$. Here we have two cases.

\case{1}{$w_j > 2^\ell$} In this case we have, again, $2^k < 2W$ and we use the same argument employed above.

\case{2}{$w_j \leq 2^\ell$} In this case we have $k = \ell + 1$. Denote with $\widetilde{H_i}$ the total height of shelves in $B_i$. Then we have $\widetilde{H_\ell} \geq T_j - H_j \geq \widetilde{T}_j - H_j \geq 6H_j$, otherwise we could pack $p_j$ in $B_\ell$. Thus, we can apply Lemma~\ref{lem:constdensity} and conclude that the box $B_\ell$ of size $2^\ell \times T_j$ is filled with constant density. Here we have two cases.

\case{2.1}{$\widetilde{H_k} \leq T_j$} In this case we have $\ALG = O(2^k T_j)$ and, thanks to the constant density packing of $B_\ell$ we have $\Sigma_j = \Theta (2^\ell \widetilde{H_\ell}) = \Theta(2^k T_j)$. Since $\OPT \geq \Sigma_j$, we get $\ALG = O(\OPT)$.

\case{2.2}{$\widetilde{H_k} > T_j$} In this case we still have $\ALG = O(2^k \widetilde{H_k})$. Moreover, $\widetilde{H_k} = O(H_n + \Sigma_n / 2^k)$, in fact if $2^{s-1} < H_n \leq 2^s$, then the total height of sparse shelves is $\sum_{i \leq s} 2^i = 2^{s+1} = O(H_n)$. Furthermore, dense shelves are filled with constant density, therefore their total height is at most $O(\Sigma_n / 2^k)$.

Finally, we need to show that $2^k = O(\min\{\sqrt[4]{\Sigma_n}, \sqrt{n}\})$. Let 
$T_j = \widetilde{T}_{j^\prime}$, we have two cases. 

\case{2.2.1}{$\Sigma_{j^\prime} < {j^\prime}^2$}
We have $\widetilde{T}_{j^\prime} \geq \Sigma_{j^\prime}^{3/4}$. And thanks to the constant density packing of $B_\ell$, we have also $2^k \Sigma_{j^\prime}^{3/4} = O(2^\ell T_j) = O(\Sigma_{j^\prime})$. Dividing both sides by $\Sigma_{j^\prime}^{3/4}$ we get $2^k = O(\sqrt[4]{\Sigma_{j^\prime}})$. 

\case{2.2.2}{$\Sigma_{j^\prime} \geq {j^\prime}^2$}
In this case we have $\widetilde{T}_{j^\prime} \geq H_{j^\prime} \sqrt{j^\prime}$. Using the constant density argument we get $2^k H_{j^\prime} \sqrt{j^\prime} = O(2^k \widetilde{T}_{j^\prime}) = O(\Sigma_{j^\prime}) \leq O(j^\prime \cdot WH_{j^\prime})$. Dividing both sides by $H_{j^\prime} \sqrt{j^\prime}$ we obtain $2^k = O(W\sqrt{j^\prime})$. Therefore, we have
\begin{equation*}
    2^k = \begin{cases}
   O(\sqrt[4]{\Sigma_{j^\prime}}) & \text{ if } \Sigma_{j^\prime} < {j^\prime}^2\\
    W\sqrt{j^\prime} & \text{otherwise.}
        \end{cases}
\end{equation*}
In the end notice that, thanks to the long edge hypotheses $H_n \leq \Sigma_n$, thus
\begin{align*}
\ALG &= O\left(2^k \widetilde{H_k}\right) = O\left(2^k H_n + \Sigma_n\right) \\ 
&= O\left(\mathbbm{1}_{\{\Sigma_{j^\prime} < {j^\prime}^2\}} \cdot H_n \sqrt[4]{\Sigma_{j^\prime}} + \mathbbm{1}_{\{\Sigma_{j^\prime} \geq {j^\prime}^2\}} \cdot WH_n \sqrt{j^\prime} + \Sigma_n\right) \\
&\leq  
\OPT \cdot O\left( \min\left\{\sqrt[4]{\Sigma_{j^\prime}}, \sqrt{j^\prime}\right\}\right) 
=
\OPT \cdot O\left( \min\left\{\sqrt[4]{\OPT}, \sqrt{n}\right\}\right).
\end{align*}
\end{proof}

\section{Further questions}

It is natural to consider problems where the given pieces are more general, such as convex polygons.
Here, we may allow the pieces to be rotated by arbitrary angles.
In that case, it follows from the technique described by Alt~\cite{DBLP:journals/eatcs/Alt16} that one can obtain a constant competitive ratio for computing a packing with a minimum perimeter bounding box:
For each new piece, we rotate the piece so that a diameter of the piece is horizontal.
We then use the algorithm $\smallboxrot$ to pack the bounding boxes of the pieces.
Since the area of each piece is at least half of the area of its bounding box, the density of the produced packing is at least half of the density of the packing of the bounding boxes.
This results in an increase of the competitive ratio by a factor of at most $\sqrt 2$.

For the problem of minimizing the perimeter of the bounding box (or convex hull) with convex polygons as pieces and only translations allowed, we do not know if it is possible to get a competitive ratio of $O(1)$, and this seems to be a very interesting question for future research.
In order to design such an algorithm, it would be sufficient to show that for some constants $\delta>0$ and $\Sigma>0$, there is an online algorithm that packs any stream of convex polygons of diameter at most $\delta$ and total area at most $\Sigma$ into the unit square, which is in itself an interesting problem.
The three-dimensional version of this question has a negative answer, even for offline algorithms:
Alt, Cheong, Park, and Scharf~\cite{alt2019packing} showed that for any $n\in\N$, there exists a finite number of 2D unit disks embedded in 3D that cannot all be packed by translation in a cube with edges of length $n$.

% Likewise, it is interesting if the gaps between the rather modest lower bounds on the competitive ratios and those realized by the algorithms for the perimeter versions described in Section~\ref{sec:perimeter} can be tightened.

% \cite{koebe1936kontaktprobleme} proved the matter!

%	\input{IntersectionApp}

	\bibliographystyle{plain}
	\bibliography{lib}

\begin{thebibliography}{10}

\bibitem{ahn2012aligning}
Hee-Kap Ahn and Otfried Cheong.
\newblock Aligning two convex figures to minimize area or perimeter.
\newblock {\em Algorithmica}, 62(1-2):464--479, 2012.

\bibitem{DBLP:journals/eatcs/Alt16}
Helmut Alt.
\newblock Computational aspects of packing problems.
\newblock {\em Bulletin of the {EATCS}}, 118, 2016.

\bibitem{alt2019packing}
Helmut Alt, Otfried Cheong, Ji-won Park, and Nadja Scharf.
\newblock Packing {2D} disks into a {3D} container.
\newblock In {\em International Workshop on Algorithms and Computation (WALCOM
  2019)}, pages 369--380, 2019.

\bibitem{altapproximating}
Helmut Alt, Mark de~Berg, and Christian Knauer.
\newblock Approximating minimum-area rectangular and convex containers for
  packing convex polygons.
\newblock In {\em 23rd Annual European Symposium on Algorithms (ESA 2015)},
  pages 25--34, 2015.

\bibitem{althurtado}
Helmut Alt and Ferran Hurtado.
\newblock Packing convex polygons into rectangular boxes.
\newblock In {\em 18th Japanese Conference on Discrete and Computational
  Geometry (JCDCGG 2000)}, pages 67--80, 2000.

\bibitem{baker1983shelf}
Brenda~S. Baker and Jerald~S. Schwarz.
\newblock Shelf algorithms for two-dimensional packing problems.
\newblock {\em SIAM Journal on Computing}, 12(3):508--525, 1983.

\bibitem{borodin2005online}
Allan Borodin and Ran El-Yaniv.
\newblock {\em Online computation and competitive analysis}.
\newblock Cambridge University Press, 2005.

\bibitem{Brubach14}
Brian Brubach.
\newblock Improved bound for online square-into-square packing.
\newblock In {\em 12th International Workshop on Approximation and Online
  Algorithms ({WAOA} 2014)}, pages 47--58, 2014.

\bibitem{CHRISTENSEN201763}
Henrik~I. Christensen, Arindam Khan, Sebastian Pokutta, and Prasad Tetali.
\newblock Approximation and online algorithms for multidimensional bin packing:
  A survey.
\newblock {\em Computer Science Review}, 24:63--79, 2017.

\bibitem{chung2019efficient}
Fan Chung and Ron Graham.
\newblock Efficient packings of unit squares in a large square.
\newblock {\em Discrete \& Computational Geometry}, 2019.

\bibitem{Csirik1998}
J{\'a}nos Csirik and Gerhard~J. Woeginger.
\newblock On-line packing and covering problems.
\newblock In Amos Fiat and Gerhard~J. Woeginger, editors, {\em Online
  Algorithms: The State of the Art}, pages 147--177. Springer, 1998.

\bibitem{erdos1975packing}
Paul Erd\H{o}s and Ron Graham.
\newblock On packing squares with equal squares.
\newblock {\em Journal of Combinatorial Theory, Series A}, 19(1):119--123,
  1975.

\bibitem{DBLP:journals/algorithmica/FeketeH17}
S{\'{a}}ndor~P. Fekete and Hella{-}Franziska Hoffmann.
\newblock Online square-into-square packing.
\newblock {\em Algorithmica}, 77(3):867--901, 2017.

\bibitem{Fiat1998}
Amos Fiat and Gerhard~J. Woeginger.
\newblock Competitive analysis of algorithms.
\newblock In Amos Fiat and Gerhard~J. Woeginger, editors, {\em Online
  Algorithms: The State of the Art}, pages 1--12. Springer, 1998.

\bibitem{januszewski1997line}
Janusz Januszewski and Marek Lassak.
\newblock On-line packing sequences of cubes in the unit cube.
\newblock {\em Geometriae Dedicata}, 67(3):285--293, 1997.

\bibitem{lassak1997linepot}
Marek Lassak.
\newblock On-line potato-sack algorithm efficient for packing into small boxes.
\newblock {\em Periodica Mathematica Hungarica}, 34(1-2):105--110, 1997.

\bibitem{leewoo}
Hyun-Chan Lee and Tony~C. Woo.
\newblock Determining in linear time the minimum area convex hull of two
  polygons.
\newblock {\em IIE Transactions}, 20(4):338--345, 1988.

\bibitem{lubachevsky2003dense}
Boris~D. Lubachevsky and Ronald~L. Graham.
\newblock Dense packings of congruent circles in rectangles with a variable
  aspect ratio.
\newblock In Boris Aronov, Saugata Basu, János Pach, and Micha Sharir,
  editors, {\em Discrete and Computational Geometry: The Goodman-Pollack
  Festschrift}, pages 633--650. 2003.

\bibitem{LUBACHEVSKY20091947}
Boris~D. Lubachevsky and Ronald~L. Graham.
\newblock Minimum perimeter rectangles that enclose congruent non-overlapping
  circles.
\newblock {\em Discrete Mathematics}, 309(8):1947--1962, 2009.

\bibitem{milenkovic1996translational}
Victor~J. Milenkovic.
\newblock Translational polygon containment and minimal enclosure using linear
  programming based restriction.
\newblock In {\em Proceedings of the twenty-eighth annual ACM symposium on
  Theory of Computing (STOC 1996)}, pages 109--118, 1996.

\bibitem{MILENKOVIC19993}
Victor~J. Milenkovic.
\newblock Rotational polygon containment and minimum enclosure using only
  robust {2D} constructions.
\newblock {\em Computational Geometry}, 13(1):3--19, 1999.

\bibitem{doi:10.1111/j.1475-3995.1999.tb00171.x}
Victor~J. Milenkovic and Karen Daniels.
\newblock Translational polygon containment and minimal enclosure using
  mathematical programming.
\newblock {\em International Transactions in Operational Research},
  6(5):525--554, 1999.

\bibitem{PARK20161}
Dongwoo Park, Sang~Won Bae, Helmut Alt, and Hee-Kap Ahn.
\newblock Bundling three convex polygons to minimize area or perimeter.
\newblock {\em Computational Geometry}, 51:1--14, 2016.

\bibitem{SPECHT201358}
E.~Specht.
\newblock High density packings of equal circles in rectangles with variable
  aspect ratio.
\newblock {\em Computers \& Operations Research}, 40(1):58 --69, 2013.

\bibitem{DBLP:journals/sigact/Stee12}
Rob van Stee.
\newblock {SIGACT} news online algorithms column 20: the power of harmony.
\newblock {\em {SIGACT} News}, 43(2):127--136, 2012.

\bibitem{DBLP:journals/sigact/Stee15}
Rob van Stee.
\newblock {SIGACT} news online algorithms column 26: Bin packing in multiple
  dimensions.
\newblock {\em {SIGACT} News}, 46(2):105--112, 2015.

\end{thebibliography}

\end{document}